\documentclass[journal]{IEEEtran}
\ifCLASSINFOpdf
\else
\fi
\usepackage{amsmath}
\usepackage{amssymb}  
\usepackage{mathtools}

\usepackage{amsthm} 

\usepackage[normalem]{ulem} 

\usepackage{color,soul}

\usepackage{tikz}
\usepackage{pgfplots}
\usetikzlibrary{positioning}
\usetikzlibrary{calc}
\usetikzlibrary{decorations.markings}
\usetikzlibrary{arrows}
\usepgfplotslibrary{colorbrewer}
\usepgfplotslibrary{patchplots}
\pgfplotsset{compat=newest}
\usepackage{pbox} 
\usepackage{relsize}

\definecolor{A}{RGB}{223,189,217}
\definecolor{B}{RGB}{223,179,217}
\definecolor{C}{RGB}{213,159,217}
\definecolor{D}{RGB}{174,134,202}
\definecolor{E}{RGB}{174,124,202}
\definecolor{F}{RGB}{174,114,202}
\definecolor{G}{RGB}{168,114,203}
\definecolor{H}{RGB}{148,114,203}
\definecolor{I}{RGB}{122,109,198}
\definecolor{J}{RGB}{102,99,188}
\definecolor{K}{RGB}{81,88,179}
\definecolor{L}{RGB}{71,70,160}
\definecolor{M}{RGB}{51,60,154}
\definecolor{N}{RGB}{35,50,124}
\definecolor{O}{RGB}{26,50,104}
\definecolor{P}{RGB}{1,1,1}

\definecolor{lilla}{HTML}{750787}

\newtheorem{prop}{Proposition}
\newtheorem{theorem}{Theorem}

\newtheorem{lem}{Lemma}
\newtheorem{assumption}{Assumption}

\newcommand{\RM}[1]{\MakeUppercase{\romannumeral #1{}}}

\DeclareMathOperator*{\Real}{Re}
\DeclareMathOperator*{\Imag}{Im}
\DeclareMathOperator*{\L2}{\mathcal{L}^2}
\DeclareMathOperator*{\R}{\mathbb{R}}
\DeclareMathOperator{\E}{\mathbb{E}}

\DeclareMathOperator{\L2}{\mathcal{L}^2}
\DeclareMathOperator{\N}{\mathcal{N}}

\begin{document}
%
\title{Sampling Strategies for Data-Driven Inference of Input-Output System Properties}
%
%
%
%

\author{Anne~Koch, 
				Jan~Maximilian~Montenbruck, 
        and~Frank~Allg\"ower
\thanks{A Koch, JM Montenbruck and F Allg\"ower are with the Institute for Systems Theory and Automatic Control, University of Stuttgart. 
This work was funded by Deutsche Forschungsgemeinschaft (DFG, German
Research Foundation) under Germany’s Excellence Strategy - EXC 2075 -
390740016. The authors thank the International Max Planck Research School
for Intelligent Systems (IMPRS-IS) for supporting Anne Koch. For correspondence, 
{\tt\small mailto:anne.koch@ist.uni-stuttgart.de}}
}

\IEEEpubid{\begin{minipage}{\textwidth}\ \\[12pt] \\ \\
         \copyright 2020 IEEE.  Personal use of this material is  permitted.  Permission from IEEE must be obtained for all other uses, in  any current or future media, including reprinting/republishing this material for advertising or promotional purposes, creating new  collective works, for resale or redistribution to servers or lists, or  reuse of any copyrighted component of this work in other works. The final version of record is available at {http://dx.doi.org/10.1109/TAC.2020.2994894}.
     \end{minipage}}
\maketitle

\begin{abstract}
Due to their relevance in controller design, we consider the problem of determining the $\mathcal{L}^2$-gain, passivity properties and conic relations of an input-output system. While, in practice, the input-output relation is often undisclosed, input-output data tuples can be sampled by performing (numerical) experiments. Hence, we present sampling strategies for discrete time and continuous time linear time-invariant systems to iteratively determine the $\mathcal{L}^2$-gain, the shortage of passivity and the cone with minimal radius that the input-output relation is confined to. These sampling strategies are based on gradient dynamical systems and saddle point flows to solve the reformulated optimization problems, where the gradients can be evaluated from only input-output data samples. This leads us to evolution equations, whose convergence properties are then discussed in continuous time and discrete time.
\end{abstract}


%
\IEEEpeerreviewmaketitle

\section{Introduction}
%
\IEEEPARstart{C}{ontrol} theory based on mathematical models is used in most existing control applications. While the theory for model-based control is quite elaborate, acquiring the mathematical model by first principles or identification of the plant can be time-consuming and highly dependent on expert knowledge. Hence, with the growing complexity of systems, acquiring the mathematical model becomes more and more challenging. 
At the same time, the amount of available data is growing rapidly due to increasing computational power and sheer size of storage. Hence, researchers from diverse backgrounds and fields are facing the challenges and chances arising from this phenomenon commonly known as big data. In recent years, this development has also attracted more and more attention in engineering applications, where data usually comprise probing input signals and probed output signals from experiments and simulations. One main question is hence, how can we best benefit from information in form of data in the state-of-the art automatic control theory?

Many existing approaches of what is called data-driven controller design are summarized in \cite{Hou2013}, which strive to learn a controller directly from data without identifying a model first. In most approaches therein, however, stability for the closed loop cannot be guaranteed or one needs to assume a certain controller structure beforehand. One complementary approach to the direct controller design from data is to learn and analyze certain system-theoretic properties from data first and leverage this knowledge to design a controller. In fact, properties such as the $\mathcal{L}^2$-gain, the shortage of passivity and conic relations of the input-output behavior allow for the direct application of well-known feedback theorems for controller design, as shown for example in \cite{Zames1966,Desoer1975}. Thus, by learning such system-theoretic properties from data, one is not bound to a controller structure beforehand and insights to the a priori unknown system are obtained. Moreover, the approach can provide control theoretic guarantees for the closed-loop behavior.

There have been different approaches to learn certain
system properties or, more generally, dissipation inequalities
from input-output data tuples that are stored and available
for analysis. In \cite{Montenbruck2016}, the authors derive overestimates on the $\mathcal{L}^2$-gain, the shortage of passivity and the cone containing all input-output samples based on finite, but densely sampled, input-output data. In \cite{Romer2017a}, this approach is extended to a more general formulation of dissipation inequalities, where the ordering of the supply rates via the S-procedure allows for inference of system properties from only finite input-output data. However, to receive quantitative bounds on certain dissipation inequalities, these approaches generally introduce conservatism to account for yet unmeasured data points in the sense that they often could not verify a system property that the system did satisfy.
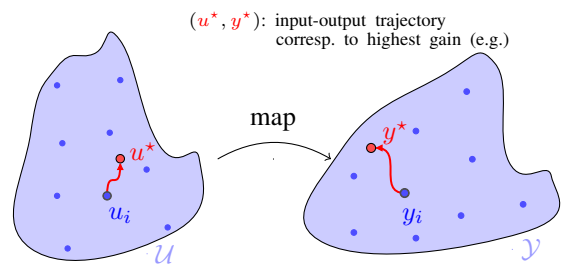
\begin{figure}[b]
 \begin{tikzpicture}[scale=0.7]
  \draw[fill=blue!20!white] plot [smooth cycle, tension=.9] coordinates {(-.7,.75) (0,1) (.5,1.25) (1.25,-.75) (2,-1) (1.5,-2.4) (-1.2,-2.7) (-1.2,-.5)};
 \draw[->] (2.35,-1) to[bend left] (4.5,-1);
 \draw (3.4,-1) -- (3.4,-1) node[anchor=south,yshift=7] {map};
  \draw[color=blue!40!white] (1,-2.8) -- (1,-2.8) node[anchor=west] {$\mathcal{U}$};	

\draw[red, line width=.75] (.25,-1.7) .. controls (.25,-1.1) and (.5,-1.7) .. (.5,-1.1);
\draw[red,fill=red] (.5,-1.1) -- (.45,-1.2) -- (.55,-1.2) -- cycle;
\draw[black, fill=red!70!white, anchor=south, yshift=3pt] (.5,-1.1) circle (.08);
\filldraw[black!70!white] (.25,-1.7) circle (.08);

\filldraw[blue!70!white] (.5,.5	) circle (.05);
\filldraw[blue!70!white] (.3,-.5) circle (.05);
\filldraw[blue!70!white] (-.6,-.7) circle (.05);
\filldraw[blue!70!white] (1,-1.2) circle (.05);
\filldraw[blue!70!white] (-.7,-1.7) circle (.05);
\filldraw[blue!70!white] (-.5,-2.7) circle (.05);
\filldraw[blue!70!white] (1.4,-2.3) circle (.05);
\filldraw[blue!70!white] (-.7,.4) circle (.05);
\filldraw[blue!70!white] (.25,-1.7) circle (.05);
\node[anchor=north, xshift=5pt, yshift=-1pt, blue] at (.25,-1.7) {$u_i$};

\node[anchor=west, xshift=0pt, yshift=5pt, red] at (.5,-1.1) {$u^\star$};

  \end{tikzpicture} \hspace{-60pt}\raisebox{0.2em}{\begin{tikzpicture}[scale=0.75]
  \draw[fill=blue!20!white] plot [smooth cycle, tension=.9] coordinates {(0,.75) (1,1) (1.25,-.75) (2,-1) (1.5,-2) (-2,-2.2) (-1.7,-.5)};
  \draw[color=blue!40!white] (1.3,-2.3) -- (1.3,-2.3) node[anchor=west] {$\mathcal{Y}$};
	
\draw[red, line width=.75] (-.6,-1.3) .. controls (-1.1,-1.3) and (-.6,-.5) .. (-1.1,-.5);
\draw[red,fill=red] (-1.1,-.5) -- (-1.0,-.55) -- (-1.0,-.45) -- cycle;
\draw[black, fill=red!70!white, anchor=west, xshift=3pt] (-1.3,-.5) circle (.08);
\filldraw[black!70!white] (-.6,-1.3) circle (.08);
	
\filldraw[blue!70!white] (.5,.5) circle (.05);
\filldraw[blue!70!white] (.6,-.7) circle (.05);
\filldraw[blue!70!white] (-.4,-.2) circle (.05);
\filldraw[blue!70!white] (1.5,-1.5) circle (.05);
\filldraw[blue!70!white] (-.6,-1.3) circle (.05);
\filldraw[blue!70!white] (-.4,-2.1) circle (.05);
\filldraw[blue!70!white] (-1.5,-1) circle (.05);
\filldraw[blue!70!white] (-1.5,-2) circle (.05);
\filldraw[blue!70!white] (.4,-1.7) circle (.05);
\node[anchor=south, xshift=-15pt, yshift=-7pt, blue] at (.25,-1.7) {$y_i$};

\node[anchor=south west, xshift=3pt, yshift=2pt, red] at (-1.3,-.7) {$y^\star$};

\node[anchor=east, xshift=-2pt, yshift=1pt, text width = 5cm] at (2.5,1.7) {\scriptsize $({\color{red}{u^\star}},{\color{red}{ y^\star}})$: input-output trajectory};
\node[anchor=east, xshift=-2pt, yshift=1pt, text width = 5cm] at (4.05,1.35) {\scriptsize corresp.\ to highest gain (e.g.)};

\end{tikzpicture}}
\caption{Given finitely many data-samples from a system that maps inputs $u \in \mathcal{U}$ to the outputs $y \in \mathcal{Y}$, where $\mathcal{U}$ is the input space and $\mathcal{Y}$ is the output space. We can determine system properties of our input-output map by finding bounds on the input-output operator through continuity assumptions (cf. \cite{Montenbruck2016}, \cite{Romer2017a}). Our research objective is to draw further data tuples $(u_i, y_i)$ to actually converge to the input-output trajectory $(u^\star, y^\star)$ that corresponds to the respective system property, e.g. the $\mathcal{L}^2$-gain.}
\label{fig:smart_sampling}
\end{figure}
If we assume, however, that additional (numerical) experiments can be performed, which is certainly the case in many application examples, then 
we can iteratively choose further data tuples in order to decrease conservatism and to obtain sharper bounds on the respective properties, as schematically illustrated in Fig.~\ref{fig:smart_sampling}. 

Instead of finding the worst-case overestimate on input-output samples, we thus want to iteratively draw further input-output samples which allow us to converge to the true $\mathcal{L}^2$-gain, shortage of passivity, or minimal radius of a cone containing the input-output behavior. 
In fact, some results in this direction can be found in \cite{Wahlberg2010} and \cite{Rojas2012}, where an iterative approach for determining the $\mathcal{L}^2$-gain of discrete time linear time-invariant (LTI) systems is presented and discussed. Building upon this idea, the concept of sampling strategies for data-driven inference of passivity properties and conic relations were introduced in our prior work \cite{Romer2017b} and \cite{Romer2018b}, respectively. In this paper, we present these approaches as a more general idea of the following form:
\subsubsection*{Sampling strategy for inference of system properties}
\begin{enumerate}
\item Formulate system property as optimization problem
\item Iteratively perform further (numerical) experiments
\begin{enumerate}
\item Calculate the gradient from input-output samples
\item Update the input along the (negative) gradient.
\end{enumerate}
\end{enumerate} 
This general approach is schematically illustrated in Fig. \ref{fig:sampling2}.
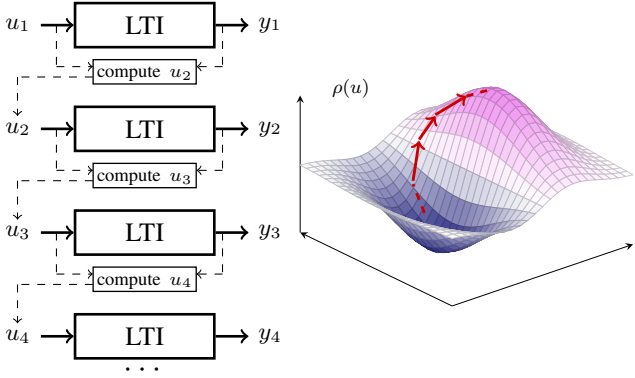
\begin{figure}
\begin{tikzpicture}[scale=.69]
\draw (0,0) node[draw,line width=1pt, inner sep=5pt,minimum size=12pt,text width=1.5cm,align=center](H) {LTI};
\draw[->, line width=1pt] (-2,0) node[left,font=\small] {$u_1$} -- (H) ;
\draw[->, line width=1pt] (H) -- (2,0) node[right,font=\small] {$y_1$};
\draw[dashed] (-1.7,0)--(-1.7,-.8);
\draw[dashed, ->] (-1.7,-.8) -- (-1,-.8);
\draw[dashed] (1.5,0)--(1.5,-.8);
\draw[dashed, ->] (1.5,-.8) -- (1,-.8);
\draw (0,-.9) node[draw,line width=.5pt, inner sep=1pt,minimum size=2pt,text width=1.3cm,text height = 0.2cm,align=center](u) {\scriptsize compute $u_2$};
\draw[dashed] (-1,-1)--(-2.45,-1);
\draw[dashed, ->] (-2.45,-1) -- (-2.45,-1.75) ;
\begin{scope}[shift={(0,-2)}]
\draw (0,0) node[draw,line width=1pt, inner sep=5pt,minimum size=12pt,text width=1.5cm,align=center](H) {LTI};
\draw[->, line width=1pt] (-2,0) node[left,font=\small] {$u_2$} -- (H) ;
\draw[->, line width=1pt] (H) -- (2,0) node[right,font=\small] {$y_2$};
\draw[dashed] (-1.7,0)--(-1.7,-.8);
\draw[dashed, ->] (-1.7,-.8) -- (-1,-.8);
\draw[dashed] (1.5,0)--(1.5,-.8);
\draw[dashed, ->] (1.5,-.8) -- (1,-.8);
\draw (0,-.9) node[draw,line width=.5pt, inner sep=1pt,minimum size=2pt,text width=1.3cm,text height = 0.2cm,align=center](u) {\scriptsize compute $u_3$};
\draw[dashed] (-1,-1)--(-2.45,-1);
\draw[dashed, ->] (-2.45,-1) -- (-2.45,-1.75) ;
\end{scope}
\begin{scope}[shift={(0,-4)}]
\draw (0,0) node[draw,line width=1pt, inner sep=5pt,minimum size=12pt,text width=1.5cm,align=center](H) {LTI};
\draw[->, line width=1pt] (-2,0) node[left,font=\small] {$u_3$} -- (H) ;
\draw[->, line width=1pt] (H) -- (2,0) node[right,font=\small] {$y_3$};
\draw[dashed] (-1.7,0)--(-1.7,-.8);
\draw[dashed, ->] (-1.7,-.8) -- (-1,-.8);
\draw[dashed] (1.5,0)--(1.5,-.8);
\draw[dashed, ->] (1.5,-.8) -- (1,-.8);
\draw (0,-.9) node[draw,line width=.5pt, inner sep=1pt,minimum size=2pt,text width=1.3cm,text height = 0.2cm,align=center](u) {\scriptsize compute $u_4$};
\draw[dashed] (-1,-1)--(-2.45,-1);
\draw[dashed, ->] (-2.45,-1) -- (-2.45,-1.75) ;
\end{scope}
\begin{scope}[shift={(0,-6)}]
\draw (0,0) node[draw,line width=1pt, inner sep=5pt,minimum size=12pt,text width=1.5cm,align=center](H) {LTI};
\draw[->, line width=1pt] (-2,0) node[left,font=\small] {$u_4$} -- (H) ;
\draw[->, line width=1pt] (H) -- (2,0) node[right,font=\small] {$y_4$};
\node[font = \large] at (0,-.65) {$\cdots$};
\end{scope}
\end{tikzpicture} 
\raisebox{3em}{\begin{tikzpicture}[scale=.9]  
\begin{axis}[
		axis lines=left, xtick=\empty, ytick=\empty, ztick=\empty,
		view ={-40}{35},
		small,
		colormap/violet,
]
\addplot3[
	surf,
	domain=-2:2,
	domain y=-1.3:1.3,
] 
	{exp(-x^2-y^2)*x};	
	\node[right, black] at (axis cs:-1.35,1.4,0.395){\footnotesize{$\rho(u)$}};	
\addplot3[dashed,red!80!black, very thick] coordinates { (0.9,0.7,0.25) (1.2,0.5,0.3)};
\addplot3[red!80!black, very thick, ->] coordinates { (0.6,1,0.11) (0.9,0.7,0.25)};
\addplot3[red!80!black, very thick, ->] coordinates {(0.5,1.2,-0.07) (0.6, 1, 0.1) };
\addplot3[red!80!black, very thick, ->] coordinates {(0.3,1.15,-0.3) (0.5,1.2,-0.08) };
\addplot3[red!80!black, very thick, dashed] coordinates {(0.1,0.8,-0.42) (0.3,1.15,-0.32) };
\end{axis}
\end{tikzpicture}}
\caption{The general idea is to iteratively perform (numerical) experiments to converge to the parameter corresponding to a certain system property, e.g. the $\mathcal{L}^2$-gain. The gradient ascent algorithm is based on only input-output data while the system and hence the optimization function $u \mapsto \rho(u)$ remains undisclosed.}
\label{fig:sampling2}
\end{figure}
\subsection{Outline}
The thrust of this work is to present a systematic approach to iteratively determine certain dissipation inequalities from input-output samples, where the input-output map remains undisclosed. In particular, we use multiple input-output trajectories with known initial condition to investigate the $\mathcal{L}^2$-gain, the shortage of passivity and conic relations, respectively.
We start in Sec.~\ref{sec:discrete} with discrete time LTI systems with a thorough analysis of continuous time optimization as well as the implications for the iterative scheme (discrete time optimization), where advanced sampling schemes can improve the convergence rate. While the general ideas have already been presented in \cite{Wahlberg2010, Romer2017b} and \cite{Romer2018b}, we are presenting stronger convergence results considering the shortage of passivity and conic relations. Furthermore, we introduce an improved iterative sampling strategy for both, the shortage of passivity and the cone with minimum radius. 
In Sec.~\ref{sec:extensions}, we generalize the framework presented in Sec.~\ref{sec:discrete}. Firstly, we extend the general approach to continuous time LTI systems in Sec.~\ref{sec:continuous} leading us to evolution equations, whose convergence properties we then investigate. We then shortly summarize the extension to multiple input multiple output (MIMO) systems in Sec.~\ref{sec:mimo} as presented in \cite{Oomen2014} and \cite{Romer2018a} and additionally provide results on the robustness of the presented framework to measurement noise. Finally, we apply the introduced approaches to different simulation examples in Sec.~\ref{sec:example} including an oscillating system and a high dimensional system, and end with a short conclusion in Sec.~\ref{sec:conclusion}.
\section{System Analysis for Discrete Time LTI Systems}
\label{sec:discrete}
Since our premise is to determine system properties from input-output data, one natural approach is the input-output framework introduced and presented for example in \cite{Zames1966} and \cite{Desoer1975}. Hence, we assume our system to be an operator that maps inputs $u$ to outputs $y$. In practical application, this input-output map is often undisclosed. However, we can perform simulations or experiments where we choose the input $u$ and measure the corresponding output $y$.
We start with a single input single output (SISO) discrete time LTI systems 
\begin{align}
y(t) = \sum_{k=0}^{\infty} g_k u(t-k),
\label{eq:1}
\end{align}
where $g_k$ denotes the impulse response sequence, $u$ is the input to the system and $y$ is the output of the system. 
For a given input sequence $u(t), t = 1, ..., n$ the input to output operator in \eqref{eq:1} can be written in matrix notation
\begin{align}
\begin{pmatrix} y(1)\\ \vdots \\ y(n) \end{pmatrix}
=
 \begin{pmatrix} g_0 & 0 & 0 & ... & 0 \\ 
 				g_1 & g_0 & 0 & ... & 0 \\ 
 				g_2 & g_1 & g_0 & ... & 0 \\ 
 				\vdots &  &  & \ddots &  \vdots \\ 
 				g_{n-1} & g_{n-2} & ... & g_1 & g_0 \\ 
 				\end{pmatrix}  				
\begin{pmatrix} u(1)\\ \vdots \\ u(n) \end{pmatrix}
\label{eq:yGu}
\end{align}
in the following denoted by $y = G u$ with $u, y \in \mathbb{R}^n$ and $G \in \mathbb{R}^{n \times n}$.
The matrix $G$ representing the convolution operator for finite inputs $u \in \mathbb{R}^n$ is a lower triangular Toeplitz matrix. 
Note that we assume $u(t)=0$ for $t \leq 0$ and only consider causal, asymptotically stable systems. However, the ideas of \cite{Tanemura2019b} can be applied for converting the sampling strategies to closed-loop approaches where pre-stabilizing controllers enable the application to unstable systems.
\subsection{$\mathcal{L}^2$-Gain}
\label{sec:gain}
The small-gain theorem, as for example presented in \cite{Zames1966}, plays an important role in systems analysis, stability studies and controller design. With the knowledge of an upper bound on the $\mathcal{L}^2$-gain of open-loop elements, the stability of the closed loop can be validated. 
The constant $\gamma$ is an upper bound on the $\mathcal{L}^2$-gain of a dynamical system if
\begin{align}
\|y\| \leq \gamma \|u\|
\label{eq:gamma}
\end{align}
holds for all input-output tuples $(u,y)$, where $u$ and $y$ are taken from some Hilbert space $\mathcal{H}$ of which $\langle \cdot, \cdot \rangle$ denotes the inner product and $\| \cdot \|$ the corresponding induced norm. A graphical interpretation of such a gain bound is depicted in Fig.~\ref{fig:l2gain}.
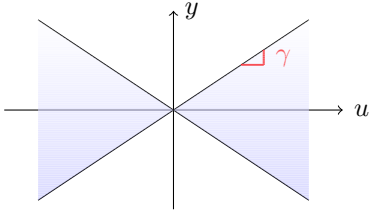
\begin{figure}[ht]
\centering
\begin{tikzpicture}[scale=0.6]
     \draw[->] (-3.75,0) -- (3.75,0) ;
		 \node[right=.8pt] at (3.75,0){$u$};
     \draw[->] (0,-2.2) -- (0,2.2) ;
		 \node[right=.8pt] at (0,2.2){$y$};
		\draw[-] (-3,-2.0) -- (3,2.0);		
		\draw[-] (-3,2.0) -- (3,-2.0);		
		\draw[-, red, thick] (1.5,1.0) -- (3*2/3,1.0);
		\draw[-, red,thick] (3*2/3,1.0) -- (3*2/3,2.0*2/3);
		\node[right=.8pt, red,thick] at (3*2/3,2.0*7/12){\textbf{$\gamma$}};
    \draw[draw=none, top color=white,bottom color=blue!30,fill opacity=0.3]
    (0,0)-- (-3,-2.0) -- (-3,2.0) -- cycle;
		 \draw[draw=none, top color=white,bottom color=blue!30,fill opacity=0.3]
    (0,0)-- (3,2.0) -- (3,-2.0) -- cycle;
\end{tikzpicture}
\caption{A graphical illustration of the $\mathcal{L}^2$-gain, denoted by $\gamma$, in the input-output plane.}
\label{fig:l2gain}
\end{figure}
For an iterative model-free approach to determine the $\mathcal{L}^2$-gain $\gamma$, we formulate the definition \eqref{eq:gamma} into the optimization problem 
\begin{align}
\gamma^2 = \sup_{\|u\|^2 \neq 0} \frac{\|y\|^2}{\|u\|^2}.
\label{eq:l2gain}
\end{align}
For discrete time LTI systems as given in \eqref{eq:yGu}, the optimization problem in \eqref{eq:l2gain} then reads
\begin{align}
\gamma^2  = \max_{\|u\|^2 \neq 0} \rho_1 (u) = \max_{\|u\|^2 \neq 0} \frac{u^\top G^\top G u}{\|u\|^2} ,
\label{eq:rayleigh1}
\end{align}
where the term $\rho_1(u)$ is also referred to as the Rayleigh quotient. The Rayleigh quotient is a smooth function $\rho_1(u): \mathbb{R}^n  \setminus \{0\} \rightarrow \mathbb{R}$ that is scale-invariant since $\rho_1(u) = \rho_1(\alpha u)$ holds for all scalars $\alpha \neq 0$. Therefore, it is sufficient to consider the Rayleigh quotient on the unit sphere $S^{n-1} =\{ u \in \mathbb{R}^n | \|u\| = 1\}$.
The critical points and critical values of $\rho_1$ are the eigenvectors and eigenvalues of $G^\top G$, respectively, as shown for example in \cite{Helmke1996}. Thus, the maximum value of the Rayleigh quotient \eqref{eq:l2gain} is exactly the maximum eigenvalue $\lambda_1$ of the symmetric matrix $G^\top G$. This relation is also referred to as the variational characterization of eigenvalues or as the Courant-Fischer-Weyl principle.

Our first proposition recasts a results of \cite{Wahlberg2010} and states that the gradient of the Rayleigh quotient can in fact be computed by only sampling two input-output data tuples, which can be generated, for example, from simulations or experiments.
\begin{prop}
The gradient vector field of $\rho_1: S^{n-1} \rightarrow \R$ 
is given by
\begin{align}
\nabla \rho_1(u) = 2 G^\top G u - 2 \rho_1(u) u
\label{eq:grad1a}
\end{align}
and can be computed by evaluating $u \mapsto Gu$ twice.
\label{prop:1a}
\end{prop}
\begin{proof}
We endow the unit sphere $S^{n-1}$ with the standard Riemannian metric, i.e. the Riemannian metric induced from the embedding $S^{n-1} \subset \mathbb{R}^n$. Hence, the gradient at $u \in S^{n-1}$ is uniquely determined by
\begin{align*}
\nabla \rho_1(u)&= 
\frac{2 G^\top G u \cdot u^\top u - 2 u^\top G^\top G u \cdot u}{(u^\top u)^2} \\
&=2 G^\top G u - 2 \rho_1(u) u.
\end{align*}
In order to compute $\nabla \rho_1(u)$ from evaluating $u \mapsto Gu$, we define the involutory permutation matrix
\begin{align*}
P = \begin{pmatrix} 0 & ... & 0 & 1 \\
0 &  ... & 1 & 0 \\
\vdots  & .^{.^{.}} & &\vdots \\
1 &  ... & 0 & 0 
\end{pmatrix}
\end{align*} 
with $P = P^{-1}$. 
Note that the matrices $G$ and $G^\top$ are involutory conjugate since $PG^\top = GP$ holds. Hence, we can compose $G^\top u$ by
${G^\top u = PGPu}$.
This finally leads to 
\begin{align*}
\nabla \rho_1(u) =  2 P G P G u - 2 (u^\top  P G P G u) u
\end{align*} 
which only consists of operations we can perform by evaluating $u \mapsto Gu$. 
\end{proof}
In experiments or simulations, the term $PGPGu$ can be obtained by performing one (numerical) experiment $y = Gu$, applying the reversed output $PGu$ to $G$ in a second experiment and reversing the output again. 
\subsubsection{$\mathcal{L}^2$-Gain - Continuous Time Solution} 
To find the maximum of the Rayleigh quotient, we employ a gradient dynamical system
\begin{align}
\frac{\mathrm{d}}{\mathrm{d} \tau} u(\tau) = \nabla \rho_1(u(\tau))
\label{eq:dyn_sys_1a}
\end{align}
with $u(\tau) = \begin{pmatrix} u(\tau,1), ..., u(\tau,n) \end{pmatrix}$, along whose solutions $\rho_1$ increases monotonically, as illustrated in Fig.~\ref{fig:rhos_illustration}(a).
This leads us to the evolution equation
\begin{align}
\frac{\mathrm{d}}{\mathrm{d} \tau} u(\tau) = 2 G^\top G u(\tau) - 2 \rho_1(u(\tau)) u(\tau), 
\label{eq:raydyn}
\end{align}
also known as the Rayleigh quotient gradient flow. 
It is readily verified that the gradient flow \eqref{eq:raydyn} leaves the sphere $S^{n-1}$ invariant \cite{Helmke1996}.
On the sphere $S^{n-1}$, the Rayleigh quotient gradient flow is in fact equivalent to the so-called \textit{Oja flow}
\begin{align}
\frac{\mathrm{d}}{\mathrm{d} \tau} u(\tau) =(u(\tau)^\top (G^\top G -  u(\tau)^\top G^\top G u(\tau)  I_n) u(\tau),
\label{eq:oja}
\end{align}
defined in $\mathbb{R}^n$, 
with $I_n$ being the $n \times n$ identity matrix. The \textit{Oja flow} is used, for example, in neural network learning theory as a means to determine the eigenvectors corresponding to the largest eigenvalues.
\begin{theorem}
Assume ${ \lambda_1 > \lambda_{2} \geq ... \geq \lambda_n }$ for the eigenvalues $\lambda_i$ of $G^\top G$. 
For almost all initial conditions $u(0)$ with $\|u(0)\| = 1$, $\rho_1$ converges to $\gamma^2$, the squared $\mathcal{L}^2$-gain, along the solutions of \eqref{eq:raydyn}. 
\label{thm:1a}
\end{theorem}
\begin{proof}
This result follows directly from Thm.~3.4 in \cite{Helmke1996}.
\end{proof}
Note that the convergence almost everywhere only excludes starting points in the union of eigenspaces of $G^\top G$ corresponding to the eigenvalues $\lambda_2,...,\lambda_n$, which is a nowhere dense subset of $S^{n-1}$.
\begin{figure}[ht]%
\centering
{
\definecolor{mycolor1}{rgb}{0.00000,0.44700,0.74100}%
\begin{tikzpicture}

\begin{axis}[%
width=0.18\textwidth,
height=0.15\textwidth,
at={(0cm,0cm)},
axis lines=left, xtick=\empty, 
scale only axis,
xmin=0,
xmax=110,
xlabel={\phantom{k} $\tau$ \phantom{k}},
xmajorgrids,
ymin=-0.6,
ymax=0.2,
ymajorgrids,
yticklabels={,,},
ylabel={$\rho_1 (u(\tau))$},
axis background/.style={fill=white}
]
\addplot [color=mycolor1,solid,forget plot]
  table[row sep=crcr]{%
1	-0.678766018796828\\
2	-0.58921067245482\\
3	-0.506379218381688\\
4	-0.430359590426493\\
5	-0.361216375120553\\
6	-0.298985551677845\\
7	-0.243666412751247\\
8	-0.195208670813367\\
9	-0.153491117480858\\
10	-0.118285254291913\\
11	-0.0891929129088353\\
12	-0.0655457606279279\\
13	-0.0462858236917084\\
14	-0.0299856596427596\\
15	-0.0152821561273365\\
16	-0.00139758373513255\\
17	0.0119289649227392\\
18	0.0247504113001466\\
19	0.03706512392231\\
20	0.0488602151676209\\
21	0.0601211741729933\\
22	0.0708337600724437\\
23	0.0809843525201968\\
24	0.0905600285537514\\
25	0.0995485956564224\\
26	0.107938617358506\\
27	0.115719436217299\\
28	0.12288119465532\\
29	0.129414853634438\\
30	0.135312209084072\\
31	0.140565905973281\\
32	0.145169449876492\\
33	0.149117215817918\\
34	0.15240445406815\\
35	0.155027292357943\\
36	0.156982733536999\\
37	0.158268646607985\\
38	0.158883745130599\\
39	0.158827487374844\\
40	0.158883606730264\\
41	0.158826065168842\\
42	0.158876285645092\\
43	0.158771995481494\\
44	0.158804915650707\\
45	0.158393913832851\\
46	0.157197501215369\\
47	0.158449124600938\\
48	0.157239245201372\\
49	0.158458101266488\\
50	0.15724615361474\\
51	0.158459613573304\\
52	0.157247321605825\\
53	0.158459870460731\\
54	0.157247520561509\\
55	0.158459914373598\\
56	0.157247554703601\\
57	0.158459921946284\\
58	0.157247560627597\\
59	0.158459923270679\\
60	0.157247561674442\\
61	0.158459923507953\\
62	0.157247561865466\\
63	0.158459923552331\\
64	0.157247561902391\\
65	0.158459923561289\\
66	0.157247561910272\\
67	0.158459923563337\\
68	0.157247561912223\\
69	0.158459923563891\\
70	0.1572475619128\\
71	0.158459923564068\\
72	0.157247561912998\\
73	0.158459923564133\\
74	0.157247561913075\\
75	0.158459923564159\\
76	0.157247561913106\\
77	0.15845992356417\\
78	0.15724756191312\\
79	0.158459923564174\\
80	0.157247561913125\\
81	0.158459923564176\\
82	0.157247561913127\\
83	0.158459923564178\\
84	0.157247561913129\\
85	0.158459923564177\\
86	0.157247561913129\\
87	0.158459923564178\\
88	0.157247561913129\\
89	0.158459923564178\\
90	0.157247561913129\\
91	0.158459923564178\\
92	0.157247561913129\\
93	0.158459923564178\\
94	0.15724756191313\\
95	0.158459923564178\\
96	0.157247561913129\\
97	0.158459923564178\\
98	0.15724756191313\\
99	0.158459923564178\\
100	0.157247561913129\\
101	0.158459923564178\\
};

\pgfplotsset{
    after end axis/.code={
        \node[above] at (axis cs:50,.2){\small{(a)}};    
    }
}

\end{axis}
\end{tikzpicture}%
}
{
\definecolor{mycolor1}{rgb}{0.00000,0.44700,0.74100}%
\begin{tikzpicture}

\begin{axis}[%
width=0.18\textwidth,
height=0.15\textwidth,
at={(0cm,0cm)},
axis lines=left, xtick=\empty, 
scale only axis,
xmin=0,
xmax=23,
xlabel={$k$},
xmajorgrids,
ymin=-0.6,
ymax=0.2,
ylabel={$\rho_1 (u_k)$},
yticklabels={,,},
ymajorgrids,
axis background/.style={fill=white}
]
\addplot [color=mycolor1,mark size=1.5pt,only marks,mark=*,mark options={solid},forget plot]
  table[row sep=crcr]{%
1	-0.6\\
2	-0.300700848219992\\
3	-0.0954488612874305\\
4	-0.0313721396074676\\
5	0.0188401384777219\\
6	0.0538288506860507\\
7	0.0776652815515371\\
8	0.0948467210811977\\
9	0.106721459247973\\
10	0.115182493587895\\
11	0.120957663463824\\
12	0.12508252737633\\
13	0.127845860093466\\
14	0.12985516357986\\
15	0.131156629499007\\
16	0.13214495721754\\
17	0.132742823366923\\
18	0.133241072314412\\
19	0.133501639986227\\
20	0.133765241465315\\
21	0.133864407586912\\
};

\pgfplotsset{
    after end axis/.code={
        \node[above] at (axis cs:10,.2){\small{(b)}};  
    }
}
\end{axis}
\end{tikzpicture}%
}
\caption{Illustration of the (a) continuous time and (b) discrete time gradient ascent optimization of the cost function $\rho_1: S^{n-1} \rightarrow \mathbb{R}$.}%
\label{fig:rhos_illustration}
\end{figure}
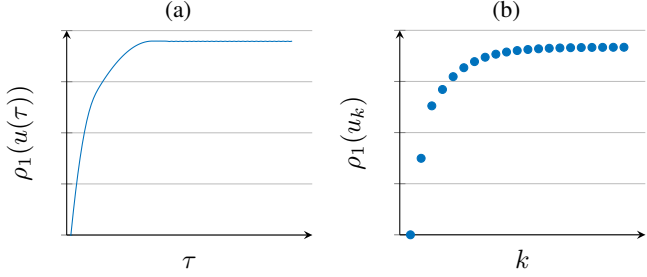
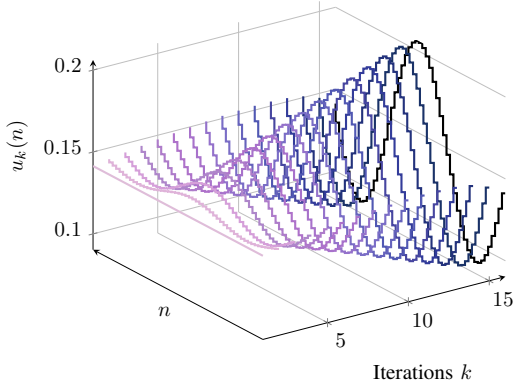
\begin{figure}[ht]
\centering
\begin{tikzpicture}[scale=0.8]
    \begin{axis}[  
		axis lines=left, ytick=\empty, 
        xlabel={Iterations $k$},
        ylabel={$n$},
        zlabel={$u_k(n)$},
        view/h=-35,
        grid=both, 
        cycle list name=color list]  
    \addplot3[const plot, no marks, color=P, line width = 1pt] table {ffigs/T16.dat};    
    \addplot3[const plot, no marks, color=O, line width = 1pt] table {ffigs/T15.dat};    
    \addplot3[const plot, no marks, color=N, line width = 1pt] table {ffigs/T14.dat};
    \addplot3[const plot, no marks, color=M, line width = 1pt] table {ffigs/T13.dat};
    \addplot3[const plot, no marks, color=L, line width = 1pt] table {ffigs/T12.dat};
    \addplot3[const plot, no marks, color=K, line width = 1pt] table {ffigs/T11.dat};
    \addplot3[const plot, no marks, color=J, line width = 1pt] table {ffigs/T10.dat};
    \addplot3[const plot, no marks, color=I, line width = 1pt] table {ffigs/T9.dat};
    \addplot3[const plot, no marks, color=H, line width = 1pt] table {ffigs/T8.dat};
    \addplot3[const plot, no marks, color=G, line width = 1pt] table {ffigs/T7.dat};    
    \addplot3[const plot, no marks, color=F, line width = 1pt] table {ffigs/T6.dat};    
    \addplot3[const plot, no marks, color=E, line width = 1pt] table {ffigs/T5.dat};    
    \addplot3[const plot, no marks, color=D, line width = 1pt] table {ffigs/T4.dat};    
    \addplot3[const plot, no marks, color=C, line width = 1pt] table {ffigs/T3.dat};    
    \addplot3[const plot, no marks, color=B, line width = 1pt] table {ffigs/T2.dat};    
    \addplot3[no marks, color=A, line width = 1pt] table {ffigs/T1.dat};    
    \end{axis}
\end{tikzpicture}
\caption{An inital input $u_0 \in \mathbb{R}^n$ iteratively converges towards the input $u$ corresponding to the parameter $\gamma$ ($\L2$-gain).}
\label{fig:PDE_disc}
\end{figure}
\subsubsection{$\L2$-Gain - Discrete Time Solution}  
In an experimental setup, we can only iteratively determine the gradient from input-output data and hence, we extend the result to discrete time optimization. Generally, an iterative approach for maximizing $\rho_1$ in \eqref{eq:rayleigh1} is to construct a sequence $(u_{k})_{k=1,2,...}$, as illustrated in Fig.~\ref{fig:PDE_disc}, such that ${ \rho_1(u_{k+1}) > \rho_1(u_{k}) }$ holds for all $k$ as depicted in Fig.~\ref{fig:rhos_illustration}(b). 
One standard tool in numerical linear algebra to construct such a sequence converging towards the dominant eigenvector and corresponding eigenvalue of a linear operator is the power method
\begin{align}
u_{k+1} = \frac{G^\top G u_k}{\| G^\top G u_k \|}.
\label{eq:power_method}
\end{align}
This method, which was amongst other methods proposed in \cite{Wahlberg2010}, presents a possible approach for iteratively estimating the $\mathcal{L}^2$-gain without any knowledge of an explicit expression of $G$, since we can retrieve the expression $G^\top G u_k$ for any $u_k \in \mathbb{R}^n$ by sampling two input-output samples $G u_k$ and $G^\top G u_k = P G P G u_k$. Hence, by iterative input-output sampling, we can apply the power method for finding the $\mathcal{L}^2$-gain while the input-output operator remains undisclosed.
\begin{prop}[\cite{Helmke1996,Wahlberg2010}]
Let the largest eigenvalue of $G^\top G$ be unique. 
Then, for almost all initial conditions $u_0$ with $\|u_0\| = 1$, the sequence $(u_k)_{k=1,2,\dots}$ constructed by \eqref{eq:power_method} converges to the dominant eigenvector of $G^\top G$ and hence, $\rho_1$ converges to $\gamma^2$, the squared $\mathcal{L}^2$-gain, along this sequence.
\end{prop}
With the scale-invariance of the Rayleigh quotient, the relevant information is contained in the direction of $u_k$. In other words, an iterate $u_k \in S^{n-1}$ represents the one-dimensional subspace $\{\beta u_k : \beta \in \mathbb{R} \}$. This links our approaches to optimizing over the real projective $(n{-}1)$-space, usually denoted by $\mathbb{RP}^{n{-}1}$. The real projective $(n{-}1)$-space is defined as the set of all lines through the origin in $\mathbb{R}^{n}$. For further information, the interested reader is referred to \cite{Helmke1996}.

To improve the convergence rate of \eqref{eq:power_method}, the application of the Lanczos method is introduced in \cite{Wahlberg2010}. In \cite{Rojas2012}, Rojas et al. propose another approach to decrease the required input-output samples. Since $G$ is a lower triangular Toeplitz matrix, the matrix $ P G $ is symmetric and can be factorized into $ Q \Lambda Q^\top$ where $Q$ is an orthonormal matrix and $ \Lambda $ is a diagonal matrix with the eigenvalues of $PG$. With $G^\top G = P G P G = Q \Lambda Q^\top Q \Lambda Q^\top = Q \Lambda^2 Q^\top$, we hence find that the maximum absolute eigenvalue of $ P G $ is exactly the square root of the maximal eigenvalue of $G^\top G$. Hence, finding the $\L2$-gain by applying the power method to the matrix $P G$ requires only one sample per iteration.

There is another well-known method that maximizes the Rayleigh quotient, namely the Rayleigh quotient iteration, which uses the iteration scheme
\begin{align} u_{k+1} = \frac{((G^\top G) - \rho_1(u_k)I_n)^{-1} u_k}{\|((G^\top G) - \rho_1(u_k)I_n)^{-1} u_k\|}.
\end{align}
The Rayleigh quotient iteration has stronger convergence properties \cite{Batterson1989}, but it cannot be applied in the present setting since we cannot compute the right hand side without explicit knowledge of $G$.
\subsection{Passivity}
\label{sec:pass}
Besides the $\L2$-gain, passivity is one of the key properties that can be exploited in order to analyze stability and design controllers, cf. \cite{Schaft2000}. The relevance of passivity for feedback control was recognized early, providing well-known feedback theorems for passive systems (cf. \cite{Zames1966} and \cite{Desoer1975}).
We start with a general input-output definition of passivity \cite{Desoer1975}. 
A system that maps inputs $u$ to the outputs $y$ is said to be passive if
\begin{align}
\langle y, u \rangle \geq 0
\label{eq:pas}
\end{align}
holds for all input-output tuples $(u,y)$, where $u$ and $y$ are taken from some Hilbert space $\mathcal{H}$. For controller design, however, we are specifically interested to which extent a system is or is not passive. 
The shortage of passivity is defined as the smallest $s$ such that
 \begin{align}
\langle y, u \rangle \geq -s \|y\|^2 
\label{eq:sop}
\end{align}
holds for all input-output tuples $(u,y)$. The system is said to be output strictly passive if $s<0$. A graphical illustration of output strict passivity is given in Fig. \ref{fig:shortage}.
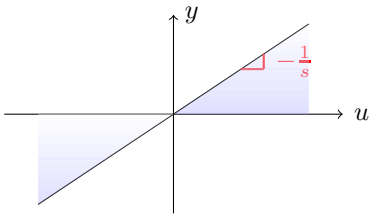
\begin{figure}[b]
\centering
\begin{tikzpicture}[scale=0.6]
     \draw[->] (-3.75,0) -- (3.75,0) ;
		 \node[right=.8pt] at (3.75,0){$u$};
     \draw[->] (0,-2.2) -- (0,2.2) ;
		 \node[right=.8pt] at (0,2.2){$y$};
		\draw[-] (-3,-2.0) -- (3,2.0);		
		\draw[-, red, thick] (1.5,1.0) -- (3*2/3,1.0);
		\draw[-, red,thick] (3*2/3,1.0) -- (3*2/3,2.0*2/3);
		\node[right=.8pt, red,thick] at (3*2/3,2.0*7/12){\textbf{$-\frac{1}{s}$}};
    \draw[draw=none, top color=white,bottom color=blue!30,fill opacity=0.3]
    (0,0)-- (-3,-2.0) -- (-3,0) -- cycle;
		 \draw[draw=none, top color=white,bottom color=blue!30,fill opacity=0.3]
    (0,0)-- (3,2.0) -- (3,0) -- cycle;
\end{tikzpicture}
\caption{A graphical illustration of output strict passivity in a plane, where $s$ denotes the shortage of passivity.}
\label{fig:shortage}
\end{figure}
For $s>0$, the shortage of passivity corresponds to the excess of passivity of a controller required to render the closed loop stable. For a more detailed description of passivity and its relevance for the application of well-known feedback theorems, the reader is referred to \cite{Zames1966}, \cite{Schaft2000} and Chapter 6 of \cite{Desoer1975}. Another parameter to determine to which extent a system is or is not (input strictly) passive is the input-feedforward passivity index, for which a sampling strategy can be found in \cite{Romer2017b}, extended in \cite{Tanemura2019}. While we only consider the shortage of passivity here in detail, input-strict passivity nicely fits into the general framework presented in this paper and the extensions and discussions in Sec.~\ref{sec:extensions} hence also hold for the input feedforward passivity parameter.

The shortage of passivity definition~\eqref{eq:sop} for discrete time LTI systems of the form \eqref{eq:yGu} reads
\begin{align}
u^\top G u \geq -s u^\top G^\top G u,
\label{eq:sop_lti}
\end{align}
which must hold for all admissible inputs $u$. Let us now assume that $g_0 \neq 0$. Consequently, the Toeplitz matrix $G$ has full rank and $G^{\top} G$ is positive definite.
Reformulating the definition of the system property \eqref{eq:sop_lti} into an optimization problem, with $u^\top G u = \frac{1}{2} u^\top (G+G^\top) u$ due to the symmetry in quadratic terms, leads to
\begin{align}
\begin{split}
s  = - \min_{\|u\| \neq 0} \rho_2(u) 
= - \min_{\|u\| \neq 0} \frac{1}{2}\frac{u^\top (G+G^\top) u}{u^\top G^\top G u}.
 \end{split}
 \label{eq:rho}
\end{align}
The term $\rho_2$ here is also referred to as the generalized Rayleigh quotient, which is a smooth function ${\rho_2: \mathbb{R}^n \backslash \{0\} \rightarrow \mathbb{R}}$. 
This minimization problem represents a generalized eigenvalue problem, where the critical points and critical values of $\rho_2$ are the generalized eigenvectors $v_i$ and generalized eigenvalues $\lambda_i$ of the pair ${(\frac{1}{2}(G+G^\top),G^\top G)}$ \cite{Romer2017b} defined by 
$$ \frac{1}{2}(G+G^\top) v_i = \lambda_i G^\top G v_i, \quad i = 1, \dots, n
.$$
Therefore, we are searching for the smallest generalized eigenvalue of the generalized eigenvalue problem denoted by ${\lambda_n}$. With $s=-\lambda_n$, this allows to infer information on passivity ($s \leq 0$), output strict passivity ($s < 0$) and the shortage of passivity ($s > 0$).

Due to the scale invariance of the generalized Rayleigh quotient $\rho_2$, we consider the optimization problem \eqref{eq:rho} on the sphere $S^{n-1}$. Our second proposition states that the gradient of the generalized Rayleigh quotient can be computed by only sampling three input-output tuples.
\begin{prop}
The gradient vector field of $\rho_2: S^{n-1} \rightarrow \R$ is given by
\begin{align}
\nabla \rho_2(u) & = \frac{1}{\|G u\|^2} ((G+G^\top)u - 2\rho_2(u) G^\top G u)
\label{eq:grad2a} \\
& = \frac{Gu + PGPu}{\|Gu\|^2} - \frac{u^\top (Gu + PGPu)}{\|Gu\|^4} (PG)^2 u. \notag
\end{align} 
and can be computed by evaluating $u \mapsto Gu$ thrice.
\label{prop:1b}
\end{prop}
\begin{proof}
This result follows directly from \cite{Romer2017b}, Lemma 2.
\end{proof}
Computing the gradient vector field \eqref{eq:grad2a} hence requires the three data samples $(u, Gu)$, $(u, PGPu)$, and $(u, (PG)^2 u)$ from three consecutive (numerical) experiments. For reasons of measurement noise, it is recommendable to calculate $\|Gu\|^2$ by $u^\top (PG)^2 u$ (cf. \cite{Wahlberg2010}).
\subsubsection{Passivity - Continuous Time Solution} 
In order to find the smallest generalized eigenvalue $\lambda_n$ and hence the shortage of passivity, we employ the gradient dynamical system
\begin{align}
\frac{\mathrm{d}}{\mathrm{d} \tau} u(\tau) &= - \nabla \rho_2(u(\tau)) \label{eq:dyn2} \\
&= \frac{1}{\|G u(\tau)\|^2} (2\rho_2(u(\tau)) G^\top G u(\tau) {-} (G{+}G^\top)u(\tau)) \notag
\end{align}
along whose solutions $\rho_2$ decreases monotonically. 
By 
\begin{align*}
\frac{\mathrm{d}}{\mathrm{d} \tau} \| u(\tau) \|^2 &= 2 u(\tau)^\top \frac{\mathrm{d}}{\mathrm{d} \tau} u(\tau) \\
&= \frac{2}{\|G u(\tau)\|^2} \left( 2 \rho_2(u(\tau)) u(\tau)^\top G^\top G u(\tau) \right. \\[-1ex] 
& \quad \quad \quad \quad \quad \quad - u(\tau)^\top (G+G^\top) u(\tau) \left.  \right) \\[1ex]
&= 2 \left( 2 \rho_2(u(\tau)) - 2 \rho_2(u(\tau)) \right) = 0
\end{align*}
we verified that \eqref{eq:dyn2} leaves the sphere $S^{n-1}$ invariant.

When discussing convergence on the unit sphere, it is important to recall that there can never be only one critical point of any vector field on the unit sphere. The Euler 
 characteristic is two for any even-dimensional sphere and zero for any odd-dimensional sphere \cite[Thm.~2.3]{Hirsch1976}. 
Hence, almost global convergence, which excludes a nowhere dense subset of $S^{n-1}$, is the strongest convergence result possible for our vector field on the unit sphere.

Gradient flows of Morse-functions or more generally of Morse-Bott functions have, with the topological restrictions, strong convergence properties on manifolds. 
We say $\rho_2: S^{n-1} \rightarrow \mathbb{R}$ is a Morse-Bott function provided the following three conditions from \cite[p.~21]{Helmke1996} are satisfied:
\begin{enumerate}
\item[a)] $\rho_2: S^{n-1} \rightarrow \mathbb{R}$ has compact sublevel sets.
\item[b)] $C(\rho_2) = \cup_{j=1}^{k} N_j$ with $N_j$ being disjoint, closed and connected submanifolds of $S^{n-1}$ and $\rho_2$ being constant on $N_j$, $j=1,\dots,k$.
\item[c)] $\ker \left( H_{\rho_2} (u) \right) = T_u N_j$, for all $u \in N_j$, $j=1,\dots,k$.
\end{enumerate}
Here, $C(\rho_2)$ denotes the set of critical points of $\rho_2$, $H_{\rho_2} (u)$ denotes the Hessian of $\rho_2$ at $u$, $\ker \left( H_{\rho_2} (u) \right)$ denotes the kernel of the Hessian of $\rho_2$ at $u$ and $T_u N_j$ is the tangent space of $N_j$ at $u$.

Due to the strong convergence properties of gradient flows of Morse-Bott functions, we show in the following that $\rho_2: S^{n-1} \rightarrow \R$ is indeed a Morse-Bott function as defined above.
\begin{lem}
The generalized Rayleigh quotient $\rho_2$ on the unit sphere $S^{n-1}$ is a Morse-Bott function.
\label{lem:2b}
\end{lem}
\begin{proof}
Condition a) of the definition of a Morse-Bott function requires that for all $c \in \mathbb{R}$ the sublevel set $\{ u \in S^{n-1} | \rho_2(u) \leq c \}$ is a compact subset of $S^{n-1}$. Since $S^{n-1}$ is compact and $\rho_2$ is continuous, this is satisfied.

The critical points $C(\rho_2)$ are all $u \in S^{n-1}$ such that
\begin{align*}
\left( \frac{1}{2}(G+G^\top) - \rho_2(u) G^\top G \right) u = 0,
\end{align*}
which are exactly the generalized eigenvectors of the matrix pair $(\frac{1}{2}(G+G^\top),G^\top G )$. All eigenvalues with geometric multiplicity one are hence isolated critical points. 

If the generalized eigenvalue $\lambda_i$ has geometric multiplicity of $m$, 
then the solution of the equation
\begin{align*}
\left(\frac{1}{2} (G+G^\top) - \lambda_i G^\top G \right) u = 0
\end{align*}
is an $m$-dimensional linear subspace of $\mathbb{R}^{n}$ and a closed connected submanifold $N_i$ of dimension $m-1$ on the unit sphere $S^{n-1}$. On this submanifold, $\rho_2(u) = \lambda_i$ for all $u \in N_i$. Therefore, condition b) is also satisfied.
Finally, we need to show that also condition c) holds. 
According to the above discussion, $T_u N_j$ is contained in $\ker \left( H_{\rho_2} (u) \right)$ and we only have to show that $\ker \left( H_{\rho_2} (u) \right) \subseteq T_u N_j$ $\forall u \in N_j$. Hence, we start by calculating the Hessian $H$ of $\rho_2$ at the critical points. 

Let $v_i \in C(\rho)$ be the generalized eigenvector corresponding to the generalized eigenvalue $\lambda_i$ of multiplicity $m$. 
Then the symmetric Hessian matrix of $\rho_2$ at $v_i$ is given by
\begin{align*}
H_{\rho_2} (v_i) 
= \frac{2}{\|G v_i\|^2} \left( \frac{1}{2}(G + G^\top) -  \lambda_i G^\top G \right).
\end{align*}
In this case, the vector $v_i$ is an element of an $(m-1)$-dimensional submanifold $N_i$, and the nullspace of the Hessian $ H_{\rho_2} (v_i) $ is exactly the eigenspace corresponding to $\lambda_i$. Let $\psi \in T_{v_i} S^{n-1}$ have a normal component to the eigenspace $N_i$, but $\psi$ still lies in the kernel of the Hessian $ H_{\rho_2} (v_i) $. Then
\begin{align*}
\lambda_i G^\top G \psi - \frac{1}{2}(G + G^\top) \psi = 0
\end{align*}
must hold and hence $(\psi, \lambda_i)$ is a solution to the generalized eigenvalue problem and $\psi \in N_i$, which leads to a contradiction. Hence, the Hessian $H_{\rho_2} (v_i)$ has full rank in any direction normal to $N_i$ at any $v_i \in N_i$. We say that every critical point of $S^{n-1}$ belongs to a nondegenerate critical submanifold. 

Altogether, we have shown that the generalized Rayleigh quotient $\rho_2$ is indeed a Morse-Bott function on the unit sphere $S^{n-1}$, which concludes the proof.
\end{proof}
With this result, we find strong convergence guarantees for the generalized Rayleigh quotient flow, summarized in the following theorem.
\begin{theorem}
Assume $\lambda_n < \lambda_{n-1} \leq \dots \leq \lambda_1$ for the generalized eigenvalues $\lambda_i$ of the matrix pair $\left( \frac{1}{2} (G + G^\top), G^\top G \right)$. For almost all initial conditions $u(0)$ with $\|u(0)\| = 1$, $\rho_2$ converges to $-s$, the shortage of passivity, along the solutions of \eqref{eq:dyn2}. 
\label{thm:1b}
\end{theorem}
\begin{proof}
We start by showing that $\rho_2$ has two minima at the eigenvector $\pm v_n$ corresponding to the smallest eigenvalue $\lambda_n$. All other critical points are saddle points or maxima of $\rho_2$.

The linearization of \eqref{eq:dyn2} on the unit sphere at any generalized eigenvector $v_i$ corresponding to the generalized eigenvalue $\lambda_i$, $i=1, \dots n$, reads
\begin{align*}
\frac{\partial}{\partial \tau} u(\tau) = \frac{2}{\|G v_i\|^2} \left( \lambda_i G^\top G - \frac{1}{2}(G+G^\top) \right) u(\tau)&, \\  u(\tau)^\top v_i = 0&.
\end{align*}
To study the exponential stability of the critical points, we are now interested in the eigenvalues of the symmetric matrix $(\lambda_i G^\top G - \frac{1}{2}(G+G^\top) )$ for all $i=1, \dots, n$ corresponding to the critical points $v_i$. Since $(\lambda_i G^\top G - \frac{1}{2}(G+G^\top) )$ is a symmetric matrix, the eigenvalues of $(\lambda_i G^\top G - \frac{1}{2}(G+G^\top) )$ in the tangent space $T_{v_i} S^{n-1}$ are all negative if and only if 
\begin{align}
u^\top \left(\lambda_i G^\top G - \frac{1}{2}(G+G^\top) \right) u < 0 \quad \forall u \in T_{v_i} S^{n-1}.
\label{eq:cond}
\end{align} 

From the definition of critical points, we know
$\lambda_i G^\top G v_i - \frac{1}{2}(G+G^\top) v_i = 0$
for all $i=1, \dots, n$. Adding $\lambda_n G^\top G v_i$ to both sides reads
\begin{align}
\lambda_n G^\top G v_i - \frac{1}{2}(G+G^\top) v_i = (\lambda_n - \lambda_i) G^\top G v_i.
\label{eq:pfthm2}
\end{align}

Since $\frac{1}{2}(G+G^\top), G^\top G$ are symmetric matrices, there exists a basis for $\R^n$ of generalized eigenvectors, which are $G^\top G$-orthogonal (i.e.~$v_i^\top G^\top G v_j = 0$, for $i \neq j$). Thus, every vector $u \in T_{v_i} S^{n-1}$ can be decomposed into a linear combination of these generalized eigenvectors $v_i$, $i=1, \dots, n$.
Multiplying $v_i^\top$ on both sides of \eqref{eq:pfthm2}, we retrieve
\begin{align*}
v_i^\top \left( \lambda_n G^\top G  - \frac{1}{2}(G+G^\top) \right) v_i = (\lambda_n - \lambda_i) \|G v_i\|^2
\end{align*}
which is strictly less than zero for all $i \neq n$. With $u = \sum_{i=1}^{n} \alpha_i v_i$, where $\alpha_i \in \R$ for $i=1, \dots, n$, we find 
\begin{align*}
u^\top \left( \lambda_n G^\top G  - \frac{1}{2}(G+G^\top) \right) u = \sum_{i=1}^{n} \alpha_i^2 (\lambda_n - \lambda_i) \|G v_i\|^2
\end{align*}
which is negative if at least one $\alpha_i \neq 0$ for any $i{=}1, \dots, n{-}1$.
With condition \eqref{eq:cond}, the eigenvalues of $(\lambda_n G^\top G - \frac{1}{2}(G+G^\top) )$ in the tangent space $T_{v_n} S^{n-1}$ are hence all negative, and the critical points $\pm v_n$ are exponentially stable.
Analogously, the definition of generalized eigenvalues yields
\begin{align*}
v_n^\top \left( \lambda_i G^\top G  - \frac{1}{2}(G+G^\top)\right) v_n = (\lambda_i - \lambda_n) \|G v_n \|^2.
\end{align*}
Since we can always choose $\alpha \in \mathbb{R}$ such that $(v_n + \alpha v_i) \in T_{v_i} S^{n-1}$, we find
\begin{align*}
&(v_n + \alpha v_i)^\top \left( \lambda_i G^\top G  - \frac{1}{2}(G+G^\top)\right) (v_n + \alpha v_i) \\
&= (\lambda_i - \lambda_n) \|G v_n \|^2,
\end{align*}
which is strictly greater than zero for all $i \neq n$. With condition \eqref{eq:cond}, there therefore exists at least one positive eigenvalue of $(\lambda_i G^\top G - \frac{1}{2}(G+G^\top) )$ on the tangent space $T_{v_i} S^{n-1}$ for all $i \neq n$. Any critical point $v_i$ with $i \neq n$ is hence a saddle point or a maximum of $\rho_2$.

Due to the reasoning above, only the isolated critical points $\pm v_n$ can be attractors for \eqref{eq:dyn2}. The union of generalized eigenspaces of the matrix pair $(\frac{1}{2} (G+G^\top), G^\top G)$ corresponding to the generalized eigenvalues $\lambda_1, \dots \lambda_{n-1}$ is a nowhere dense subset in $S^{n-1}$. With Lemma \ref{lem:2b}, $\rho_2$ is a Morse-Bott function and Prop.~3.9 from \cite{Helmke1996} applies. Therefore, every solution of the gradient flow converges as $t \rightarrow \infty$ to an equilibrium point, and hence, every solution of \eqref{eq:dyn2} starting in the complement of the union of generalized eigenspaces corresponding to the generalized eigenvalues $\lambda_1, \dots \lambda_{n-1}$ will converge to either $v_n$ or $-v_n$. This completes our proof.
\end{proof}
One other approach to investigate the generalized Rayleigh quotient $\rho_2$ is by performing a linear coordinate transform $y=Gu$. Since $G$ is full rank, there always exists an inverse transformation. Define ${T=G^{-\top} (G+G^\top) G^{-1}}$. By transformation, we retrieve the standard Rayleigh quotient
\begin{align*}
\rho_2(u) 
&= \frac{1}{2}\frac{u^\top G^{T} G^{-\top}(G+G^\top)G^{-1} G u}{u^\top G^\top G u} = \frac{1}{2}\frac{y^\top T y}{y^\top y}
\end{align*} 
with the symmetric matrix $T$, where the eigenvalues of $T$ correspond to the generalized eigenvalues $\lambda_i$ of the pair ${ (\frac{1}{2}(G + G^\top), G^\top G) }$. With $G$ unknown, however, this transformation cannot directly be used for an iterative scheme to determine the shortage of passivity.
\subsubsection{Passivity - Discrete Time Solution}  
In any application, we can only iteratively determine the gradient. We thus extend our results to discrete time optimization where we improve the convergence through exact line search. Generally speaking, discrete time minimization problems on manifolds can be approached by the general update formula
\begin{align} u_{k+1} = R_{u_k} (\alpha_k p_k)
\label{eq:update_general}
\end{align}
where the search direction $p_k$ lies in the tangent space 
$T_{u_k}S^{n-1}$ and $\alpha_k$ denotes the step length. The mapping $R_{u_k}$ is also called a retraction mapping from the tangent space $T_{u_k}S^{n-1}$ to the manifold $S^{n-1}$ \cite{Absil2008}. Choosing
\begin{align}
u_{k+1} &= R_{u_k} (\alpha_k p_k)  = \frac{u_k+\alpha_k p_k}{\|u_k+\alpha_k p_k\|}
\label{eq:update}
\end{align}
yields a valid retraction onto the sphere $S^{n-1}$ \cite{Absil2008}, which is defined for all vectors that lie in a tangent space $T_{u}S^{n-1}$.

Let the search direction be the negative gradient $p_k = -\nabla \rho_2(u_k) \in T_{u_k}S^{n-1}$, which can be computed from data tuples according to Prop.~\ref{prop:1b}. 
There exist various approaches on how to choose the step size $\alpha_k$. Literature on this topic has a long history and goes back to \cite{Hestenes1951,Hestenes1951b}, where the convergence for (generalized) Rayleigh quotient iterations with fixed step size or optimized step sizes are investigated. In fact, even though the input-output map of the discrete time LTI system remains undisclosed, we can still perform a line search algorithm in the present setting. Minimizing
\begin{align}
& 2 \rho_2(R_{u_k} (\alpha_k p_k)) = 2 \rho_2(u_k + \alpha_k p_k) \notag \\
&= \frac{u_k^\top (G {+} G^\top) u_k + 2 \alpha_k u_k^\top (G {+} G^\top) p_k + \alpha_k^2 p_k^\top (G {+} G^\top) p_k}{u_k^\top G^\top G u_k + 2 \alpha_k u_k^\top G^\top G p_k + \alpha_k^2 p_k^\top G^\top G p_k} \notag \\
&= \frac{\begin{pmatrix} 1 \\ \alpha_k \end{pmatrix}^\top \begin{pmatrix} u_k^{\mathsmaller{\top}} (G{+}G^{\mathsmaller{\top}}) u_k & u_k^{\mathsmaller{\top}} (G{+}G^{\mathsmaller{\top}}) p_k \\ u_k^{\mathsmaller{\top}} (G{+}G^{\mathsmaller{\top}}) p_k & p_k^{\mathsmaller{\top}} (G{+}G^{\mathsmaller{\top}}) p_k \end{pmatrix} \begin{pmatrix} 1 \\ \alpha_k \end{pmatrix}}
{\begin{pmatrix} 1 \\ \alpha_k \end{pmatrix}^\top \begin{pmatrix} u_k^{\mathsmaller{\top}} G^{\mathsmaller{\top}} G u_k & u_k^{\mathsmaller{\top}} G^{\mathsmaller{\top}} G p_k \\ u_k^{\mathsmaller{\top}} G^{\mathsmaller{\top}} G p_k & p_k^{\mathsmaller{\top}} G^{\mathsmaller{\top}} G p_k \end{pmatrix} \begin{pmatrix} 1 \\ \alpha_k \end{pmatrix}}
\label{eq:linesearch}
\end{align} 
with respect to the step size $\alpha_k$ yields yet another generalized eigenvalue problem. Scaling the eigenvector that corresponds to the smaller eigenvalue such that the first entry equals one, the second entry denotes the optimized step size $\alpha_k^\star$. The optimized step size can again be computed by evaluating $u \mapsto Gu$ three additional times, without knowledge of $G$.
We generalize in the following the main result from \cite{Tanemura2019}, which shows for the input-feedforward passivity index that no additional input-output samples are required for finding the optimal step size. Similarly this also holds for the shortage of passivity via induction.
\begin{theorem}
Given $(G{+}G^\top) u_k$, $G^\top G u_k$, $(G{+}G^\top) p_k$, $G^\top G p_k$ and $\alpha_k^\star$. Then the gradient $p_{k+1}$ and the optimal step size $\alpha_{k+1}^\star$ can be computed by evaluating $u \mapsto Gu$ thrice.
\end{theorem}
\begin{proof}
Since $(G+G^\top) u_{k+1} = (G+G^\top) u_{k} + \alpha_k^\star (G+G^\top) p_k$ and $G^\top G u_{k+1} = G^\top G u_{k} + \alpha_k^\star G^\top G p_k$ holds, $p_{k+1}$ can be computed without additional input-output tuples. With the additional data tuples $(p_{k+1}, G p_{k+1})$, $(P p_{k+1}, G P p_{k+1})$ and $(PG p_{k+1}, G PG p_{k+1})$ we can calculate the optimal step size $\alpha_{k+1}^\star$ via \eqref{eq:linesearch} and at the same time also fulfill the requirement to apply this theorem at step $k+1$.
\end{proof}
\subsection{Conic Relations}
\label{sec:conic}
In \cite{Zames1966}, G. Zames introduces a feedback theorem on conic relations, which can be seen as a generalization of the small-gain theorem. In practice, an open-loop gain of less than one is often quite restrictive. With a linear shift in the feedback equation, however, a reduced gain product can often be obtained. This results in Zames' Theorem which says that the closed-loop is bounded if the open loop can be factored into two, suitably proportioned, conic relations \cite{Zames1966}.

A system that maps inputs $u$ to outputs $y$ is said to be confined to a conic region defined by the real constants $c$ and $r \geq 0$ if the inequality
\begin{align}
\| y - cu  \| \leq r \|u\| 
\label{eq:conic_relation}
\end{align}
is satisfied for all input-output tuples $(u,y)$, where $u$ and $y$ are lying in some Hilbert space $\mathcal{H}$. 
The constant $c$ is also called the center parameter and the constant $r$ is also called the radius of the input-output map. In Fig. \ref{fig:cone_plane}, a graphical interpretation of such a conic sector in the plane is depicted.
\begin{figure}[b]
\centering
\begin{tikzpicture}[scale=0.6]
     \draw[->] (-3.75,0) -- (3.75,0) ;
		 \node[right=.8pt] at (3.75,0){$u$};
     \draw[->] (0,-3) -- (0,3.75) ;
		 \node[right=.8pt] at (0,3.75){$y$};
		\draw[-] (-3.2*6/6,-1.0*6/6) -- (3.2*7/6,1.0*7/6);		
		\draw[-] (3.2*7/12,1.0*7/12) -- (3.2*7/12+0.3,1.0*7/12);
		\draw[-] (3.2*7/12+0.3,1.0*7/12) -- (3.2*7/12+0.3,1.0*7/12+0.1);
		 \node[right=.8pt] at (3.2*7/12+0.3,1.0*7/12+0.05){$c{-}r$};
     \draw[-] (-1.2*5/6,-3.8*5/6) -- (1.2,3.8);
     \draw[-] (1.2/2,3.8/2) -- (1.2/2+0.3,3.8/2);
     \draw[-] (1.2/2+0.3,3.8/2) -- (1.2/2+0.3,3.8/2+0.95);
		 \node[right=.5pt] at (1.2/2+0.3,3.8/2+0.475){$c{+}r$};
    \draw[draw=none, top color=white,bottom color=blue!30,fill opacity=0.3]
    (0,0)-- (3.2*7/6,1.0*7/6) -- (1.2,3.8) -- cycle;
		 \draw[draw=none, top color=white,bottom color=blue!30,fill opacity=0.3]
    (0,0)-- (-3.2*6/6,-1.0*6/6) -- (-1.2*5/6,-3.8*5/6) -- cycle;
		\draw[-, line width = 1pt, red] (-4.2/2 ,-3.8*5/12-1/2) -- (3.2*7/12+0.6,1.0*7/12+1.9);
		\draw[-, line width = .5pt, red] (3.2*7/24+0.3,1.0*7/24+.95) -- (3.2*7/24+0.3+0.3,1.0*7/24+.95);
		\draw[-, line width = .5pt, red] (3.2*7/24+0.3+0.3,1.0*7/24+.95) -- (3.2*7/24+0.3+0.3,1.0*7/24+.95+.3);
		\node[red, right=.5pt] at (3.2*7/24+0.3+0.3,1.0*7/24+.95+.15){$c$};
\end{tikzpicture}
\caption{A graphical illustration of a conic sector in a plane \cite{Zames1966}, which is described by the center parameter $c$ and the radius $r$.}
\label{fig:cone_plane}
\end{figure}
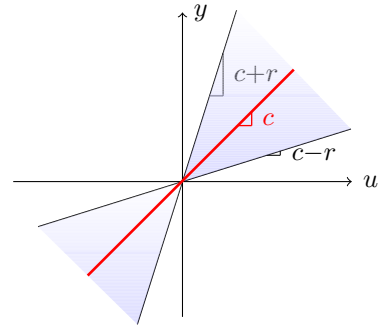 
Reformulating \eqref{eq:conic_relation} into an optimization problem yields
\begin{align*}
r^2 = \sup_{\|u\|^2 \neq 0} \; {c^2 + \frac{\|y\|^2 - 2 c \langle u,y \rangle}{\|u\|^2}}.
\end{align*}
With this maximization problem, we can find a valid radius $r$ corresponding to any center $c$ describing a cone that our input-output map is confined to.
However, the goal is to find the transformation $\pm cI$ that minimizes the gain of the open-loop element, and hence, to find the minimum radius $r_{\min}$. Finding $r_{\min}$ can increase the set of controllers for which the closed-loop is bounded. Equivalently, minimizing the radius can offer higher robustness measures for a given stabilizing controller related to the gap metric. As presented in \cite{Sakkary1985} and \cite{Georgiou1990}, the gap between the cones, to which the open-loop elements are confined, can be interpreted as a robustness measure. 

Searching for the minimum radius $r_{\min}$ leads to 
\begin{align}
r_{\min}^2 = \inf_c \sup_{\| u \| \neq 0} \frac{\|y\|^2 - 2 c \langle u,y \rangle + c^2 \|u\|^2}{\|u\|^2},
\label{eq:conic_allg}
\end{align}
which is a min-max optimization problem in the variables $c$ and $u$.
 Applying the standard Euclidean inner product with the input-output system description from \eqref{eq:yGu}, we can rewrite the optimization problem \eqref{eq:conic_allg} into
\begin{align}
\begin{split}
r_{\min}^2 & = \min_c \max_{\|u\| \neq 0} \rho_3(c,u)  \\
& = \min_c \max_{\|u\| \neq 0} \frac{u^\top (G^\top G - c (G + G^\top) + c^2 I_n) u}{\|u\|^2}.
\end{split}
\label{eq:opt_prob}
\end{align}
The term $\rho_3 : \mathbb{R} \times \mathbb{R}^n \backslash \{0\} \rightarrow \mathbb{R}$ in \eqref{eq:opt_prob} with $A(c) = G^\top G - c (G + G^\top) + c^2 I_n$ can again be referred to as a Rayleigh quotient. The Rayleigh quotient $\rho_3$ is a smooth function, which is scale-invariant in $u$. Therefore, we consider the Rayleigh quotient $\rho_3$ on the manifold $\mathbb{R} \times S^{n-1}$. 

Due to the Courant-Fischer-Weyl principle, all critical points and critical values of $\rho_3$ for any given $c$ are the eigenvectors and eigenvalues of $A(c)$ respectively. More specifically, the maximum of the Rayleigh quotient for any given $c$ corresponds to the largest eigenvalue $\lambda_1$ of $A(c)$. 
Hence, Eq. \eqref{eq:opt_prob} could also be expressed as
\begin{align}
r_{\min}^2 = \min_c \lambda_1 (A(c)).
\label{eq:eigvec}
\end{align}
This reveals that we are searching for the minimization of the maximal eigenvalue of a symmetric matrix function. Our first result states, that \eqref{eq:eigvec} is a strongly convex function with exactly one minimum.
\begin{lem}
The function $f_1: \mathbb{R} \rightarrow \mathbb{R}$ with $c \mapsto \lambda_1 ( A(c))$ is strongly convex and has only one minimum, which is a global minimum.
\label{lem:1c}
\end{lem}
\begin{proof}
We start by separating $f_1(c) = 
\lambda_1 (G^\top G - c (G + G^\top) ) + c^2$.
While $c^2$ is a strongly convex function, we are further interested in $\lambda_1 (G^\top G {-} c (G {+} G^\top) )$. From \cite{Fletcher1985} Thm.~A.1, the largest eigenvalue is convex and continuous in the space of symmetric matrices. Hence, the largest eigenvalue of an affine function of symmetric matrices $\lambda_1(G^\top G {-} c (G {+} G^\top) )$ is convex.
Since $c^2$ is a smooth and strongly convex function and $\lambda_1 (G^\top G - c (G + G^\top) )$ is continuous and convex, their sum 
is strongly convex and continuous. 
Hence, the function $f_1: \mathbb{R} \rightarrow \mathbb{R}$ with $c \mapsto \lambda_1 ( A(c))$ has a global minimizer and no other local minima.
\end{proof}
To find the minimal radius $r_{\min}$ from our min-max optimization problem \eqref{eq:opt_prob} without knowledge of $G$, our approach is again to apply a gradient-based optimization scheme. Therefore, our first proposition states that we can indeed retrieve the gradient of $\rho_3$ with respect to $c$ and $u$ from drawing input-output samples $(u,y)$ from simulations or experiments.
\begin{prop}
The gradients of $\rho_3: \mathbb{R} \times S^{n-1} \rightarrow \R$ in the first and second variable are given by 
\begin{align}
\nabla_c \rho_3(c,u) &= 2c - \frac{u^\top (G + G^\top)  u }{ \|u\|^2} 
\label{eq:conic_gradient1} \\
&= 2c - u^\top (G u + P G P u)  \notag \\
\nabla_u \rho_3(c,u) &= \frac{2}{\|u\|^2} (A(c)  - \rho_3(c,u) I_n) u \label{eq:conic_gradient2} \\
&= 2 ( P G P G u - c (G u + P G P u) ) \notag \\
& \quad - 2 u^\top (P G P G u - c (G u + P G P u) ) u, \notag
\end{align}
and can be computed by evaluating $u \mapsto Gu$ thrice.
\label{prop:2c}
\end{prop}
\begin{proof}
This result follows directly from \cite{Romer2018b}, Lemma 1.
\end{proof} 
\subsubsection{Conic Relations - Continuous Time Solution}  
To find the conic relation with minimal radius $r_{\min}$ for our unknown input-output system, we employ a gradient descent in the first variable $c$ and a gradient ascent in the second vector variable $u$ resulting in the saddle point dynamics given by 
\begin{align}
\begin{split}
\frac{\mathrm{d}}{\mathrm{d} \tau} c(\tau) = &{\;-} \nabla_c \rho_3(c(\tau),u(\tau))\\
\frac{\mathrm{d}}{\mathrm{d} \tau} u(\tau) = & \phantom{\;-} \nabla_u \rho_3(c(\tau),u(\tau)).
\end{split}
\label{eq:conic_flow}
\end{align}
The saddle point dynamics in \eqref{eq:conic_flow} leave the manifold $\mathbb{R} \times S^{n-1}$ invariant, since
\begin{align*}
\frac{\mathrm{d}}{\mathrm{d} \tau} \| u(\tau) \|^2 &= 2 u(\tau)^\top \frac{\mathrm{d}}{\mathrm{d} \tau} u(\tau) \\
&= \frac{4}{\|u\|^2} u(\tau)^\top (A(c(\tau))  - \rho_3(c(\tau),u(\tau)) I_n) u(\tau) \\
&= 4 (\rho_3(c(\tau),u(\tau))  - \rho_3(c(\tau),u(\tau))) = 0.
\end{align*}
The equilibrium points of \eqref{eq:conic_flow} are described by $u$ being an eigenvector $v_i$ of $A(c)$ and the corresponding $c = \frac{1}{2} u (G + G^\top) u$. With the analysis before, we are searching for $u^{\star}$ being the eigenvector corresponding to the maximum eigenvalue denoted by $u^{\star} = v_1 (A(c^\star))$ with $c^\star = \frac{1}{2} u^{\star \top} (G + G^\top) u^\star$, which then leads to the minimal radius $r^2_{\min} = \rho_3 ( c^\star, u^\star)$. 

In the following theorem, we show that the tuple with the center $c^\star$ and input sample $u^\star$ corresponding to the minimal radius $r_{\min}$ is in fact a locally attracting equilibrium point of \eqref{eq:conic_flow} when $\lambda_1(A(c^\star))$ is a simple eigenvalue. When we minimize the maximal eigenvalue of a matrix function, however, we also need to consider the possibility that the solution to this optimization problem is an eigenvalue of geometric multiplicity two. In this second case, let $v_1(A(c^\star))$ and $v_2(A(c^\star))$ be the eigenvectors to the eigenvalue $\lambda_1(A(c^\star))$. We choose $v_1(A(c^\star))$ such that $2c - v_1(A(c))^\top  (G+G^\top) v_1(A(c)) = 0$ and $v_2(A(c^\star))$ consecutively such that $v_2(A(c^\star))^\top v_1(A(c^\star)) = 0$, which is always possible since $A(c)$ is a symmetric matrix. We formulate an additional assumption in the case that $\lambda_1(A(c^\star))$ is an eigenvalue of multiplicity two.
\begin{assumption}
\label{as:1}
With $v_1(A(c^\star))$ and $v_2(A(c^\star))$ as defined above, the following condition holds:
\begin{align*}
v_1(A(c^\star))^\top (G + G^\top) v_2(A(c^\star)) \neq 0.
\end{align*}
\end{assumption}
We will briefly discuss this technical assumption after the following theorem.
\begin{theorem}
Assume that $\lambda_1(A(c^\star))$ is an eigenvalue
\begin{itemize}
\item with multiplicity one, or 
\item with multiplicity two and Assumption \ref{as:1} holds. 
\end{itemize}
Then the equilibrium point $(c^\star, u^\star)$ corresponding to the squared minimum radius $r_{\min}^2 = \rho_3 (c^\star, u^\star)$ is 
locally exponentially stable under the saddle point dynamics in \eqref{eq:conic_flow}. 
\label{thm:1c}
\end{theorem}
\begin{proof}
Linearizing \eqref{eq:conic_flow} around the critical point $(c^\star, u^\star) \in \mathbb{R} \times S^{n-1}$ yields the linearized system dynamics 
\begin{align*}
\frac{\mathrm{d}}{\mathrm{d} \tau} \begin{pmatrix} \delta c(\tau) \\ \delta u(\tau) 
\end{pmatrix} = J(c^\star, u^\star) \begin{pmatrix} \delta c(\tau) \\ \delta u(\tau) 
\end{pmatrix}
\end{align*}
with $\delta c = c - c^\star$, $\delta u = u - u^\star$, and the Jacobian reads 
\begin{align*}
J (c^\star, u^\star) = \begin{pmatrix}
- \nabla_{cc} \rho_3(c^\star, u^\star) & -\nabla_{cu} \rho_3(c^\star, u^\star) \\
\phantom{-} \nabla_{uc} \rho_3(c^\star, u^\star) & \phantom{-} \nabla_{uu} \rho_3(c^\star, u^\star) 
\end{pmatrix}
\end{align*}
where
\begin{align*}
- \nabla_{cc} \rho_3(c^\star, u^\star) &= -2 \\
-\nabla_{cu} \rho_3(c^\star, u^\star) &= 2\left( (G + G^\top)u^\star - (u^{\star \top} (G + G^\top)  u^\star) u^\star \right)^\top\\
\nabla_{uc} \rho_3 (c^\star, u^\star) &= \nabla_{cu} \rho_3 (c^\star, u^\star)^\top\\
\nabla_{uu} \rho_3 (c^\star, u^\star) &= 2 \left( A(c^\star) - \rho_3(c^\star, u^\star) I_n \right).
\end{align*}

Since any symmetric matrix possesses $n$ mutually orthogonal eigenvectors, 
the set of eigenvectors of the symmetric matrix $\frac{1}{2} \left( J(c^\star, u^\star) + J (c^\star, u^\star)^\top \right)$ given by $b_1 = (1,0_{n})$, $b_2 = (0,v_n)$, $b_3 = (0,v_{n-1})$, $\dots$, $b_{n} = (0,v_2)$, $b_{n+1} = (0,v_1)$ form an orthonormal basis of $\mathbb{R}^{n+1}$, where $v_i, i=1,\dots,n$ denote the eigenvectors of $A(c^\star)$. If $\lambda_1(A(c^\star))$ has multiplicity two, we choose the two eigenvectors spanning the eigenspace corresponding to $\lambda_1(A(c^\star))$ such that $v_1 = u^\star$ and $v_1^\top v_2 = 0$. 

Recall that \eqref{eq:conic_flow} leaves the manifold $\mathbb{R} \times S^{n-1}$ invariant and hence, we are only interested in the Jacobian on the tangent space $T_{(c^\star, u^\star)} \left( \mathbb{R} \times S^{n-1} \right)$, which is spanned by the basis vectors $b_1$, $b_2$, $\dots$, $b_{n}$. 
By projecting the Jacobian matrix $J(c^\star, u^\star)$ onto the tangent space $T_{(c^\star, u^\star)} \left( \mathbb{R} \times S^{n-1} \right)$, we find
\begin{align*}
J^\prime(c^\star, u^\star) =
\begin{pmatrix}
b_1 & \dots & b_{n}
\end{pmatrix}^\top
J(c^\star, u^\star)
\begin{pmatrix}
b_1 & \dots & b_{n}
\end{pmatrix} \\
= \begin{pmatrix}
-2 & * 						 & * && \dots & * \\
* & (\lambda_n {-} \lambda_1) & 0 && \dots & 0 \\
* &  0 & \ddots &&  & 0 \\
\vdots & \vdots && \ddots &  & \vdots \\
* & 0 &&  &  (\lambda_3 {-} \lambda_1) & 0 \\
* & 0 & & \dots & 0 & (\lambda_2 {-} \lambda_1) \\
\end{pmatrix}
\end{align*}
which is of the general form
\begin{align*}
J^\prime(c^\star, u^\star) =
\begin{pmatrix}
\phantom{-}N\phantom{^\top} & S \\ 
-S^\top & C
\end{pmatrix}
\end{align*}
with $S = (2 v_1 ^\top (G + G^\top) v_n, \dots, 2 v_1 ^\top (G + G^\top) v_2 )$ and $N$ negative definite. The matrix $C$ is negative definite if the eigenvalue $\lambda_1(A(c^\star))$ is simple and negative semi-definite otherwise.

In the case that $\lambda_1(A(c^\star))$ is simple, choosing $P = I_{n}$ in
\begin{align*}
\begin{pmatrix}
N\phantom{^\top} & -S \\ S^\top & \phantom{-}C
\end{pmatrix} P + P 
\begin{pmatrix}
\phantom{-}N\phantom{^\top} & S \\ 
-S^\top & C
\end{pmatrix} 
= 2 
\begin{pmatrix}
N & 0 \\ 0 & C
\end{pmatrix}
\end{align*}
yields a negative-definite matrix. Applying Lyapunov's theory for linear systems and the Hartman-Grobman theorem, this proves local exponential stability of $(c^\star, u^\star)$ in the case that $\lambda_1(A(c^\star))$ is simple.

In the case that $\lambda_1(A(c^\star))$ has multiplicity of two, we continue by applying the \textit{Ky Fan} inequality (\cite{Bhatia1997}, Prop.~\RM{3}.5.3) to find that
$$ \operatorname{Re}\left( \lambda_i(J^\prime (c^\star, u^\star)) \right) \leq \lambda_1 \left( \frac{1}{2} \left( J^\prime(c^\star, u^\star) + J^\prime(c^\star, u^\star)^\top \right) \right).  $$
holds for all $i=1,\dots,n$.
Since  $\frac{1}{2} \left( J^\prime(c^\star, u^\star) + J^\prime(c^\star, u^\star)^\top \right)$ denotes the symmetric part of the Jacobian on the tangent space $T_{(c^\star, u^\star)} \left( \mathbb{R} \times S^{n-1} \right)$, which is negative semidefinite, we know that 
$\operatorname{Re}\left( \lambda_i(J^\prime(c^\star, u^\star)) \right) \leq 0$
for all $i=1,\dots,n$.

Furthermore, we need to exclude possible eigenvalues on the imaginary axis. Ostrowski and Schneider, in \cite{Ostrowski1962} Thm.~2, draw a connection between purely imaginary eigenvalues of a matrix $J^\prime(c^\star, u^\star)$ and conditions on its symmetric part $\frac{1}{2} \left(J^\prime(c^\star, u^\star) + J^\prime(c^\star, u^\star)^\top \right)$. Namely, if $\frac{1}{2} \left(J^\prime(c^\star, u^\star)+ J^\prime(c^\star, u^\star)^\top \right)$ is semidefinite and real, then the corresponding eigenvectors to $k=2m$ imaginary eigenvalues ($\pm i \alpha_1, \dots, \pm i \alpha_m$) of $J^\prime(c^\star, u^\star)$ are in the nullspace of $\frac{1}{2} \left(J^\prime(c^\star, u^\star) + J^\prime(c^\star, u^\star)^\top \right)$.

If $\lambda_1(A(c^\star))$ has the multiplicity of two, we find that the symmetric part of the $J^\prime(c^\star, u^\star)$, which reads
\begin{align}
\begin{split}
& \frac{1}{2} \left( J^\prime(c^\star, u^\star) + J^\prime(c^\star, u^\star)^\top \right) \\ &=
\mbox{diag} \left( -2, \lambda_n-\lambda_1, \dots, \lambda_3-\lambda_1, 0 \right),
\end{split}
\label{eq:symmetric_Jacobian}
\end{align}
has only a one dimensional nullspace. Since any eigenvalues on the imaginary axis would correspond to an even number of eigenvectors that must lie in the nullspace of \eqref{eq:symmetric_Jacobian}, we can conclude that $J^\prime(c^\star, u^\star)$ has no purely imaginary eigenvalues.  

Finally, we need to investigate possible zero eigenvalues.
Rearranging rows of $J^\prime(c^\star, u^\star)$ yields
\begin{align*}
\begin{pmatrix}
-a & 0 & 0 & & \dots  & 0 \\
* & (\lambda_n {-} \lambda_1) & 0 && \dots & 0 \\
\vdots &  & \ddots&  & & \vdots \\
* &  &&  &  (\lambda_3 {-} \lambda_1) & 0 \\
-2 & * 						 & * && \dots & a \\
\end{pmatrix} 
\end{align*}
with $a = v_2^\top (G+G^\top) u^\star$, which is a triangular matrix with non-zero entries on the diagonal if and only if $v_2^\top (G+G^\top) u^\star \neq 0$. This reveals that $J^\prime(c^\star, u^\star)$ has full rank and therefore no zero eigenvalue under Assumption 1. 

In summary, the linearization of the dynamics \eqref{eq:conic_flow} on the manifold $\mathbb{R} \times S^{n-1}$ at the equilibrium point $(c^\star, u^\star)$ lead to a Jacobian $J^\prime(c^\star, u^\star)$ with $\operatorname{Re} \left( \lambda_i (J^\prime(c^\star, u^\star) ) \right) < 0$. Hence, in the tangent space $T_{(c^\star, u^\star)} \left( \mathbb{R} \times S^{n-1} \right)$, the point $(c^\star, u^\star)$ is locally exponentially stable. This concludes our proof.
\end{proof}
Let us further consider the case when the eigenvalue $\lambda_1(A(c^\star))$ is an eigenvalue of multiplicity two and Assumption 1 does not hold, i.e. $v_2^\top (G+G^\top) u^\star = 0$. This happens only if at least one of the two analytic eigenvalue functions $\tilde{\lambda}_{i=1,2}(c)$, from rearrangement of $\lambda_{i=1,2} (A(c))$ \cite{Rellich1969}, that meet at $(c^\star,\lambda_1(A(c^\star))$ has a vanishing gradient at $c^\star$. Since this an incredibly rare case, and Assumption 1 holds almost surely when $\lambda_1(A(c^\star))$ is an eigenvalue of multiplicity two, we do not want to go into more detail here and refer the interested reader to \cite{Rellich1969, Mengi2014} for more details on eigenvalues of Hermitian matrix functions.

Furthermore, even in the technical case when Assumption~1 is not satisfied,
we find that the Jacobian with regard to the directions normal to $v_1, v_2$, given by $b_1, \dots, b_{n-1}$, has again only eigenvalues with negative real parts, since
\begin{align*}
J^{\prime \prime} (c^\star, u^\star) &=
\begin{pmatrix}
b_1 & \dots & b_{n-1}
\end{pmatrix}^\top
J(c^\star, u^\star)
\begin{pmatrix}
b_1 & \dots & b_{n-1}
\end{pmatrix} \\
&= \begin{pmatrix}
-2\phantom{^\top} & S \\
-S^\top & \mbox{diag} \left( \lambda_n {-} \lambda_1, \dots, \lambda_3 {-} \lambda_1 \right) \\
\end{pmatrix}.
\end{align*}

As a corollary to Thm.~\ref{thm:1c}, we can show that the optimizer $(c^\star, u^\star)$ is a local min-max saddle point of $\rho_3$ via the Taylor series expansion given by
\begin{align*}
\rho_3(c,u^\star) &= \rho_3(c^\star,u^\star) + \nabla_{c} \rho_3 (c^\star, u^\star) \delta c \\
&+ \frac{1}{2} \nabla_{cc} \rho_3 (c^\star, u^\star) \delta c ^2 + \mathcal{O}(\delta c ^3),\\
\rho_3(c^\star,u) &= \rho_3(c^\star,u^\star) + \nabla_{u} \rho_3 (c^\star, u^\star) \delta u \\
&+ \frac{1}{2} \delta u ^\top \nabla_{uu} \rho_3 (c^\star, u^\star) \delta u + \mathcal{O}(\delta u ^3).
\end{align*}
From before, we know $\nabla_{u} \rho_3 (c^\star, u^\star) = \nabla_{c} \rho_3 (c^\star, u^\star) = 0$. With $\nabla_{uu} \rho_3 (c^\star, u^\star)$ positive semi-definite and $\nabla_{cc} \rho_3 (c^\star, u^\star)$ negative definite on the tangent space $T_{(c^\star, u^\star)} \left( \mathbb{R} \times S^{n-1} \right)$, the inequality
$\rho_3(c^\star,u) \leq \rho_3(c^\star,u^\star) \leq \rho_3(c,u^\star)$
holds in a neighborhood of $(c^\star,u^\star)$, and hence, $(c^\star, u^\star)$ is in fact a locally exponentially stable local min-max saddle point of $\rho_3$ on the manifold $\mathbb{R} \times S^{n-1}$. 

The structure of our Jacobian $J(c^\star, u^\star)$ relates the linearization of \eqref{eq:conic_flow} around the critical point $(c^\star, u^\star)$ to the well-known (linear) saddle point problems of the general form
\begin{align*}
\begin{pmatrix}
N & \phantom{-}S^\top \\ S & -C\phantom{^\top}
\end{pmatrix}
\begin{pmatrix}
x \\ y
\end{pmatrix}=
\begin{pmatrix}
f \\ g
\end{pmatrix}
\end{align*}
that arise, for example, in the context of regularized weighted least-squares problems, from certain interior point methods in optimization, or from Lagrange functions with $C$ being a zero matrix \cite{Benzi2005}. 
\subsubsection{Conic Relations - Discrete Time Solution}
Iterative approaches for the solution of saddle point problems have already been introduced in the book of Arrow, Hurwicz and Uzawa \cite{Arrow1958} and an article of Polyak \cite{Polyak1970}. In these references, iterative schemes consisting of simultaneous iterations in both variables and their convergence are dis\-cussed, addressing mainly the problem of finding the saddle point of a Lagrangian.  
One of the iterative approaches introduced in \cite{Arrow1958}, Chapter 10, Sections 4-5, is the so-called \textit{Arrow-Hurwicz} iteration which reads
\begin{align}
\begin{split}
c_{k+1} &= c_k - \alpha  \nabla_c \rho_3(c_k, u_k) \\
u_{k+1} &= u_k + \alpha  \nabla_u \rho_3(c_k, u_k).
\end{split}
\label{eq:arrow_hurwicz_it}
\end{align}
In \cite{Romer2018b}, it is shown along the lines of \cite{Polyak1970} that for a small enough step size $\alpha$, the method \eqref{eq:arrow_hurwicz_it} is locally convergent to $(c^\star, u^\star)$ and the modified \textit{Arrow-Hurwicz} method \cite{Popov1980} is introduced as an expedient method to determine the minimal cone of an unknown input-output system.

The \textit{Uzawa} iteration for general saddle point problems, also called the dual method, was presented by Uzawa in \cite{Arrow1958}, Chapter 10. Here, the gradient iteration is only performed with respect to the input $u$, while the corresponding center $c$ is found by minimization of $\rho_3(c,u_k)$ with respect to $c$:
\begin{align}
\begin{split}
\rho_3(c_{k+1}, u_k) &= \min_{c} \rho_3(c, u_k) \\
u_{k+1} &= u_k + \alpha  \nabla_u \rho_3(c_k, u_k).
\end{split}
\label{eq:uzawa_it}
\end{align}

For any given $u_k \in S^{n-1}$, $\rho_3(\cdot, u_k)$ is a strongly convex function with a global minimum at the critical point $c = 0.5 (u_k^\top (G + G^\top) u_k)$ (cf. Lemma~\ref{lem:1c}).
With the results from Prop.~\ref{prop:2c}, we can hence compute $\min_{c} \rho_3(c, u_k) = 0.5 u_k^\top (PGP u_k + G u_k)$ with two input-output tuples and $\nabla_u \rho_3(c_k, u_k)$ with one additional input-output tuple from (numerical) experiments.
Hence, the iterative \textit{Uzawa} iteration on the manifold $\mathbb{R} \times S^{n-1}$ is given by
\begin{align}
\begin{split}
c_{k+1} &= 0.5 u_k^\top (PGP u_k + G u_k) \\
u^\prime_{k+1} &= u_k + 2 \alpha ( P G P G u_k - c (G u_k + P G P u_k)) \\
& - 2 \alpha u_k^\top (P G P G u_k - c (G u_k + P G P u_k) ) u_k\\
u_{k+1} &= \frac{u^\prime_{k+1}}{\|u^\prime_{k+1}\|},
\end{split}
\label{eq:uzawa_iteration}
\end{align}
with step size $\alpha$, where we applied again a valid retraction mapping from the tangent space $T_{(c_k,u_k)} (\R \times S^{n-1})$ to the manifold $\R \times S^{n-1}$ \cite{Absil2008}.

Along the lines of \cite{Polyak1970}, we show in the following that \eqref{eq:uzawa_iteration} is locally convergent to $(c^\star, u^\star)$.
\begin{prop}
\label{prop:1c}
Assume that $\lambda_1(A(c^\star))$ is an eigenvalue
\begin{itemize}
\item with multiplicity one, or 
\item with multiplicity two and Assumption \ref{as:1} holds. 
\end{itemize}
Then, there exists an $\bar{\alpha}$ such that for all $\alpha \in (0, \bar{\alpha})$ the method \eqref{eq:uzawa_it} is locally convergent to $(c^\star, u^\star)$. 
\end{prop}
\begin{proof}
The local behavior of the \textit{Uzawa} iteration is given by
\begin{align*}
e_{k+1} =
K(c^\star,u^\star)
e_{k}, \quad 
e_k = 
\begin{pmatrix}
c_{k}\\
u_{k} 
\end{pmatrix} -
\begin{pmatrix}
c^\star\\
u^\star 
\end{pmatrix}
\end{align*}
where
\begin{align*}
K(c^\star,u^\star) =
\begin{pmatrix}
0 & - \frac{1}{2} \nabla_{cu} \rho_3(c^\star, u^\star)\\
\alpha \nabla_{uc} \rho_3(c^\star, u^\star) & I_n + \alpha \nabla_{uu} \rho_3(c^\star, u^\star)
\end{pmatrix}.
\end{align*}
To improve readability, we denote $\nabla_{uc} \rho_3(c^\star, u^\star)$ by $S$ in the following.

Since the projection onto the manifold $\R \times S^{n-1}$ is a smooth retraction mapping from any tangent space onto the manifold, this projection preserves convergence properties of the algorithm \cite[Chapter 4]{Absil2008}. Therefore, we are interested in the eigenvalues of $K(c^\star, u^\star)$ on the tangent space $T_{(c^\star,u^\star)} (\R \times S^{n-1})$ spanned by the vectors $b_1, \dots, b_n$.

The eigenvalue equation for the matrix $K(c^\star, u^\star)$ with the eigenvalue $\mu$ and the eigenvector $(c_e, u_e)$ reads
\begin{align*}
- \frac{1}{2} S^\top u_e &= \mu c_e \\
\alpha S c_e + (I_n + \alpha \nabla_{uu} \rho (c^\star,u^\star)) u_e &= \mu u_e, \\
u_e ^\top u^\star &= 0.
\end{align*}
Multiplying the first equation by $\alpha S$ and replacing $\alpha S c_e$ by the second equation yields
\begin{align*}
- & \frac{\alpha}{2} S S^\top u_e = \mu ((\mu-1) I_n - \alpha \nabla_{uu} \rho (c^\star,u^\star) ) u_e
\end{align*}
and hence
\begin{align*}
& \frac{\alpha}{2} \left( S S^\top - 2 \nabla_{uu} \rho (c^\star,u^\star) \right) u_e = \mu (1-\mu) u_e.
\end{align*}

If $u_e = 0$, then $\mu c_e = 0$. Since $(c_e,u_e) \neq 0$, and hence $c_e \neq 0$, we find $\mu = 0$ implying $| \mu | < 1$.

If $u_e \neq 0$, then $\mu (1-\mu)$ is an eigenvalue of the matrix $\frac{\alpha}{2} \left( S S^\top - 2 \nabla_{uu} \rho (c^\star,u^\star) \right)$, which is symmetric and positive definite in the tangent space $T_{(c^\star,u^\star)} (\R \times S^{n-1})$ if $\lambda_1(A(c^\star))$ is an eigenvalue of multiplicity one, or an eigenvalue of multiplicity two and Assumption 1 holds (cf. Thm.~3). Hence, we know that $\mu (1-\mu)$ is real and
\begin{align}
0 < \mu (1-\mu) &\leq \alpha \left\| \frac{1}{2}S S^\top - \nabla_{uu} \rho (c^\star,u^\star) \right\| 
\label{eq:smallerSStop}
\\ &\leq  \frac{\alpha}{2} \| S S^\top \| + 2 \alpha \left( \lambda_1(A(c^\star))- \lambda_n(A(c^\star)) \right). \notag
\end{align}

The term $\mu (1-\mu)$ being real implies $\Imag (\mu) = 0$, or $\Real (\mu) = \frac{1}{2}$. With $\Imag (\mu) = 0$ and $\mu (1-\mu) > 0$, it follows directly that $| \mu | < 1$ must hold. 

Let 
\begin{align*}
\bar{\alpha} = \frac{1}{2 \| S S^\top - 2 \nabla_{uu} \rho (c^\star,u^\star) \| }.
\end{align*}
 
In the case $\Imag (\mu) \neq 0$ and hence $\Real (\mu) = \frac{1}{2}$, we find $\mu (1-\mu) = \frac{1}{4} + \Imag^2(\mu) > \frac{1}{4} = \bar{\alpha} \left\|  \frac{1}{2}S S^\top - \nabla_{uu} \rho (c^\star,u^\star) \right\|$. This, however, contradicts \eqref{eq:smallerSStop} for all $\alpha < \bar{\alpha}$, and thus we have $\Imag (\mu) = 0$.

Altogether, this leaves us with eigenvalues $\mu$ with an absolute value strictly less than one whenever $\alpha < \bar{\alpha}$, and hence $(e_k) \rightarrow 0 $ whenever $e_0$ is small. Therefore, the method \eqref{eq:arrow_hurwicz_it} is locally convergent to $(c^\star, u^\star)$.
\end{proof}
With this iterative approach for saddle point problems, we conclude this section of methods for discrete time linear systems to identify the $\mathcal{L}^2$-gain, the shortage of passivity and the minimal cone that the input-output system is confined to. The presented analysis of possible approaches for data-driven inference of control theoretic system properties depicts the potential of this framework and builds the basis for extensions for robust gradient based methods or application to other classes of systems. 
\section{Generalizations and Extensions}
\label{sec:extensions}
The previous section introduces a systematic approach to iteratively determine certain dissipation inequalities from input-output data and, moreover, provides a rigorous mathematical framework and hence also the foundation for generalizations and extensions. In this section, 
we start by introducing the necessary tools to evaluate also continuous time systems via iterative methods and show how a similar approach than in the previous section is applicable. Furthermore, we summarize how the presented results can also be applied to MIMO systems, discuss how measurement noise impacts the approach and present some insights into the convergence rate.
\subsection{Continuous Time LTI Systems}
\label{sec:continuous}
In this section, we consider SISO continuous time LTI systems $H: \mathcal{L}^2 \rightarrow \mathcal{L}^2$, where the set $\mathcal{L}^2$ denotes the square integrable functions. 
The input to output operator can be written as
\begin{align*}
y(t) = (g*u)(t) = \int \! g(t-\zeta) u(\zeta) \, \mathrm{d} \zeta 
\end{align*}
where $g$ denotes the continuous time impulse response of the system, $u$ is the input to the system and $y$ is the output of the system. 
In the following, the convolution operator $u \mapsto g * u$ will be denoted by $u \mapsto C_g(u)$ for readability. 
Again, we assume $u(t)=0$ for $t < 0$.

For every bounded linear operator on a Hilbert space ${H: \mathcal{L}^2 \rightarrow \mathcal{L}^2}$, 
there exists a unique adjoint operator ${H^\star: \mathcal{L}^2 \rightarrow \mathcal{L}^2}$ defined by 
$\langle H(u), y \rangle = \langle u, H^\star(y) \rangle$,
wherein $\|\cdot \| : \L2 \rightarrow \R$ denotes the $\L2$-norm and $\langle \cdot,\cdot \rangle: \L2 \times \L2 \rightarrow \R$ denotes the $\L2$-inner product. Let $\bar{g}(t) = g(-t)$.
Then
\begin{align*}
\langle C_g(u), y \rangle 
&= \int \! \int \! g(t-\zeta) u(\zeta) \, \mathrm{d}\zeta \, y(t) \, \mathrm{d}t \\
&= \int \! \int \! g(t-\zeta) y(t) u(\zeta) \, \mathrm{d}t  \, \mathrm{d}\zeta 
\\
&= \int \! u(\zeta) \int \! g(-(\zeta-t)) y(t) \, \mathrm{d}t  \, \mathrm{d}\zeta
= \langle u , C_{\bar{g}}(y) \rangle
\end{align*} 
verifies that $C_{\bar{g}}$ is the adjoint operator of $C_g$.

In the following, we iteratively search for the input ${u \in \mathcal{L}^2}$ corresponding to the operator gain $\gamma$, the shortage of passivity $s$ and conic relations, respectively, as depicted in Fig. \ref{fig:PDE_cont}.
\begin{figure}[t]
\centering
\begin{tikzpicture}[scale=0.8]
    \begin{axis}[   
		axis lines=left, ytick=\empty,
        xlabel={Iterations $k$},
        ylabel={$t$},
        zlabel={$u_k(t)$},
        view/h=-35,
        grid=both,  
        cycle list name=color list]  
    \addplot3[no marks, color=P, line width = 1pt] table {ffigs/T16.dat};    
    \addplot3[no marks, color=O, line width = 1pt] table {ffigs/T15.dat};    
    \addplot3[no marks, color=N, line width = 1pt] table {ffigs/T14.dat};
    \addplot3[no marks, color=M, line width = 1pt] table {ffigs/T13.dat};
    \addplot3[no marks, color=L, line width = 1pt] table {ffigs/T12.dat};
    \addplot3[no marks, color=K, line width = 1pt] table {ffigs/T11.dat};
    \addplot3[no marks, color=J, line width = 1pt] table {ffigs/T10.dat};
    \addplot3[no marks, color=I, line width = 1pt] table {ffigs/T9.dat};
    \addplot3[no marks, color=H, line width = 1pt] table {ffigs/T8.dat};
    \addplot3[no marks, color=G, line width = 1pt] table {ffigs/T7.dat};    
    \addplot3[no marks, color=F, line width = 1pt] table {ffigs/T6.dat};    
    \addplot3[no marks, color=E, line width = 1pt] table {ffigs/T5.dat};    
    \addplot3[no marks, color=D, line width = 1pt] table {ffigs/T4.dat};    
    \addplot3[no marks, color=C, line width = 1pt] table {ffigs/T3.dat};    
    \addplot3[no marks, color=B, line width = 1pt] table {ffigs/T2.dat};    
    \addplot3[no marks, color=A, line width = 1pt] table {ffigs/T1.dat};    
    \end{axis}
\end{tikzpicture}
\caption{The goal is to iteratively converge from an initial input $u_0 \in \mathcal{L}^2$ towards the input $u$ corresponding to the operator gain $\gamma$, the shortage of passivity $s$ or conic relations, respectively.}
\label{fig:PDE_cont}
\end{figure}
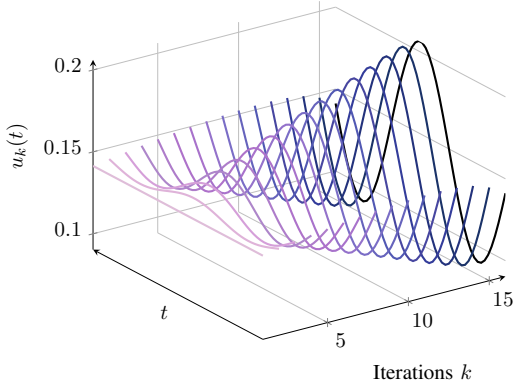
\subsubsection{$\mathcal{L}^2$-Gain}
For continuous time LTI systems we can reformulate the $\mathcal{L}^2$-gain condition into
\begin{align*}
\begin{split}
\gamma^2 &= \sup_{ \|u\| \neq 0} \rho_1(u) = \sup_{ \|u\| \neq 0} \frac{1}{2} \frac{\langle C_g(u), C_g(u)\rangle}{\|u\|^2} \\
&= \sup_{ \|u\| \neq 0} \frac{\int \! \left( \int \! g(\zeta) u(t-\zeta) \, \mathrm{d}\zeta  \, \right)^2 \, \mathrm{d}t}
{\int \! u^2(t) \, \mathrm{d}t}
\end{split}
\end{align*}
where $\rho_1: \mathcal{L}^2 \rightarrow \mathbb{R}$ is a scale-invariant function, also referred to as the Rayleigh quotient. 

We say that $\rho_1$ is Fr\'echet-differentiable in $\mathcal{L}^2$ if for every $u \in \mathcal{L}^2$ there exists a linear operator $D\rho_1(u)$ such that
\begin{align*}
\lim_{h \rightarrow 0} \frac{| \rho_1(u+h)-\rho_1(u) - D\rho_1(u)(h) |}{\|h\|} = 0.
\end{align*}
By the Riesz representation theorem, there is a unique ${\rho_1':\mathcal{L}^2 \rightarrow \mathcal{L}^2}$ such that $D\rho_1(u)(h) =  \langle \rho_1'(u), h \rangle$ if $\rho_1$ is differentiable at $u$. $D\rho_1(u)$ is also called the dual of $\rho_1'(u)$.
Without loss of generality, we consider $\rho_1$ on the unit sphere $S_{\mathcal{L}^2} = \{ u \in \mathcal{L}^2 | \|u\| = 1 \}$.
\begin{prop}
The Fr\'echet-derivative of $\rho_1$ on the unit sphere $S_{\mathcal{L}^2}$
is given by
\begin{align*}
\rho_1'(u) = 2 \left( C_{\bar{g}} \circ C_g \right) (u) - 2\rho_1(u)  u .
\end{align*}
and can be computed by evaluating $u \mapsto C_g(u)$ twice.
\label{prop:2a}
\end{prop}
\begin{proof}
First, we claim that the Fr\'echet-derivative of 
\begin{align*}
f(u) = \langle C_g(u) , C_g(u) \rangle
\end{align*}
is given by $f'(u) = 2 \left( C_{\bar{g}} \circ C_g \right) (u)$. By applying the presented definition
\begin{align*}
		| f&(u+h)-f(u)- \langle f'(u), h \rangle | \\
		=& | \langle C_g(u+h), C_g(u+h) \rangle \\
		&-  \langle C_g(u),C_g(u) \rangle 
		- 2 \langle \left( C_{\bar{g}} \circ C_g \right) (u), h \rangle|\\
		= &|\langle C_g(h), C_g(h) \rangle|  = \mathcal{O}(\|h\|^2)
		\end{align*}
we find that the claim is indeed true.

From the quotient rule follows
\begin{align*}
\rho_1'(u) &= \frac{2}{\|u\|^2} \left( \left( C_{\bar{g}} \circ C_g \right) (u) - \frac{\langle C_g(u), C_g(u) \rangle}{\langle u, u \rangle}  u \right)\\
&=2 \left( C_{\bar{g}} \circ C_g \right) (u) - 2\rho_1(u)  u.
\end{align*}

In order to calculate the Fr\'echet-derivative $\rho_1'(u)$ from two input-output tuples, we require $C_g(u)$ and $C_{\bar{g}} \circ C_g (u)$ from evaluating $u \mapsto C_g(u)$.
Therefore, we rewrite the adjoint operator 
\begin{align*}
C_{\bar{g}}(y)(t) &= \int \! g(-\zeta) y(t-\zeta) \, \mathrm{d}\zeta = \int \! g(\zeta) y(t+\zeta) \, \mathrm{d}\zeta 
\\
&= \int \! g(\zeta) \bar{y}(-t-\zeta) \, \mathrm{d}\zeta 
= C_{g}(\bar{y})(-t)
\end{align*}
to see that $\left( C_{\bar{g}} \circ C_g \right) (u) = C_{\bar{g}}(y) = \overline{C_{g}(\bar{y})} = \overline{C_{g}(\overline{C_g(u)})}$ holds, where the bar again denotes time-reversal. 
\end{proof}
Hence, even though we have no knowledge about the system but input-output information, we can construct the Fr\'echet-derivative $\rho_1'(u)$ from two input-output samples, namely $C_g(u)$ and $\left( C_{\bar{g}} \circ C_g \right) (u) = \overline{C_{g}(\overline{C_g(u)})}$ by time-reversing the output ${\overline{C_g(u)}(t) = C_g(u)(-t)}$, choosing it as yet another input and time-reversing the output again.

Similar to the discrete time case, we plan to maximize $\rho_1$ by a dynamical system that is now described by the evolution equation $\frac{\partial}{\partial \tau} u(\tau) {=} \rho_1'(u(\tau))$ reading
\begin{align}
\frac{\partial}{\partial \tau} u(\tau) {=} 2 \left( C_{\bar{g}} \circ C_g \right) (u(\tau)) {-} 2\rho_1(u(\tau))  u(\tau)  
\label{eq:PDE} 
\end{align}
along whose solution $\rho_1$ increases monotonically.
We can show that \eqref{eq:PDE} leaves the unit sphere $S_{\mathcal{L}^2}$ invariant:
\begin{align*}
\begin{split}
\frac{\partial}{\partial \tau} \|u(\tau)\|^2
=& 2\langle u(\tau), \frac{\partial}{\partial \tau} u(\tau) \rangle \\
=&4 \langle u(\tau), \left( C_{\bar{g}} \circ C_g \right) (u(\tau)) - \rho_1(u(\tau))  u(\tau) \rangle \\
=& 4 \rho_1(u(\tau)) - 4\rho_1(u(\tau))) = 0.
\end{split}
\end{align*}
\begin{prop}
The Rayleigh quotient $\rho_1(u(\tau))$ monotonically increases along the solutions of \eqref{eq:PDE} and converges for ${\tau \rightarrow \infty}$.
\label{prop:2a2}
\end{prop}
\begin{proof}
According to the Courant-Fischer-Weyl principle for self-adjoint operators we find an upper bound on $\rho_1$ by
\begin{align*}
\sup \sigma = \sup_{\|u\| \neq 0} \rho_1(u)
\end{align*}
where $\sigma$ denotes the spectrum of the linear and bounded operator $u \mapsto C_g(u)$.
This principle is also referred to as the Rayleigh-Ritz principle.
Moreover, on the basis of the Fr\'echet-derivative of Prop.~\ref{prop:2a}, we can conclude that $\rho_1(u(\tau))$ is a monotonically increasing function of $\tau$:
\begin{align}
\begin{split}
\frac{\partial}{\partial \tau} \rho_1(u(\tau))
{=}& \langle \rho_1'(u(\tau)), \frac{\partial}{\partial \tau} u(\tau) \rangle 
{=} \| \rho_1'(u(\tau)) \|^2 \geq 0.
\label{eq:2a_mono_decr}
\end{split}
\end{align}
Thus, $\rho_1(u(\tau))$ is monotonically increasing with $\tau$ and upper-bounded by the Rayleigh-Ritz principle stated above. By the monotone convergence theorem, $\tau \mapsto \rho_1(u(\tau))$ converges. 
\end{proof} 
Another promising approach of maximizing $\rho_1$ could be to apply Temple's inequality where an additional term can guarantee a lower bound on the infimum of $-\rho_1$ \cite{Harrell1978} and hence the supremum of $\rho_1$.
\subsubsection{Passivity}
Similarly to the $\mathcal{L}^2$-gain, the shortage of passivity can be studied by 
\begin{align*}
\begin{split}
s &= -\inf_{ \|u\| \neq 0} \rho_2(u) = -\inf_{ \|u\| \neq 0} \frac{\langle u, C_g(u) \rangle}{\|C_g(u)\|^2}.
\end{split}
\end{align*}
Since $\langle u, C_g(u) \rangle = \langle C_g(u), u \rangle = \langle u , C_{\bar{g}}(u) \rangle$ holds by the definition of the adjoint operator, we rewrite $\rho_2$ into
\begin{align}
\rho_2(u) = \frac{1}{2} \frac{\langle u, C_g(u) + C_{\bar{g}}(u) \rangle}{\|C_g(u)\|^2}.
\label{eq:spassive_continuous}
\end{align}
where $\rho_2: \mathcal{L}^2 \rightarrow \mathbb{R}$ is a scale-invariant function also referred to as the generalized Rayleigh quotient.
Without loss of generality, we consider $\rho_2$ on the unit sphere $S_{\mathcal{L}^2}$.
\begin{prop}
The Fr\'echet-derivative of $\rho_2$ on the unit sphere $S_{\mathcal{L}^2}$
is given by
\begin{align*}
\rho_2'(u) = \frac{1}{\| C_g(u) \|^2} \left( C_{\bar{g}}(u) {+} C_g(u) - 2 \rho_2(u) (C_{\bar{g}} \circ C_g)(u) \right)
\end{align*}
and can be computed by evaluating $u \mapsto C_g(u)$ thrice.
\label{prop:3b}
\end{prop}
\begin{proof}
The proof is analogous to the proof of Prop.~\ref{prop:2a}.
\end{proof}
To minimize $\rho_2$, we again consider the evolution equation
\begin{align}
\frac{\partial}{\partial \tau} u(\tau) = - \rho_2'(u(\tau)) \label{eq:PDE_2}
\end{align}
along whose solution $\rho_2$ decreases monotonically.
We can show that \eqref{eq:PDE_2} leaves the unit sphere $S_{\mathcal{L}^2}$ invariant.
\begin{prop}
The generalized Rayleigh quotient $\rho_2$ monotonically decreases along the solutions of \eqref{eq:PDE_2}. 
\label{prop:2b}
\end{prop}
\begin{proof}
On the basis of the Fr\'echet-derivative of Prop.~\ref{prop:2b}, we can conclude that $\rho_2(u(\tau))$ is a monotonically decreasing function of $\tau$, cf. Eq. \eqref{eq:2a_mono_decr}.
\end{proof} 
In case of conic relations, the gradients with respect to both variables $c \in \mathbb{R}$ and $u \in S_{\mathcal{L}^2}$ can also be obtained from 
evaluating $u \mapsto C_g(u)$ thrice for determining the cone with infimal radius that a continuous time system is confined to, analogously to the cases of gain and passivity as introduced before. More details are left out due to brevity.
\subsection{Multiple Input Multiple Output Systems}
\label{sec:mimo}
In this subsection, we shortly summarize how the presented approach can be extended to MIMO systems along the lines of \cite{Oomen2014, Romer2018a}. 
For simplicity we consider square MIMO systems, for which the input-output map for a given input sequence in matrix notation reads
\begin{align}
\begin{pmatrix}
y_{1} \\
y_{2} \\
\vdots\\
y_{m} \\
\end{pmatrix}
=
\begin{pmatrix}
G_{11} & G_{12} & \cdots & G_{1m} \\
G_{21} & G_{22} & \cdots & G_{2m} \\
\vdots & \vdots & \vdots &  \vdots\\
G_{m1} & G_{m2} & \cdots & G_{mm} \\
\end{pmatrix}
\begin{pmatrix}
u_{1} \\
u_{2} \\
\vdots\\
u_{m} \\
\end{pmatrix}
\label{eq:mimo}
\end{align}
with $u_i, y_i \in \mathbb{R}^n$, $i=1,\dots,m$. In short notation, \eqref{eq:mimo} will be denoted by $Y = \Gamma U$ where $Y \in \mathbb{R}^{mn}$, $U \in \mathbb{R}^{mn}$ and $ \Gamma \in \mathbb{R}^{mn \times mn}$. The input-output system properties can then again be formulated as optimization problems
\begin{align*}
\gamma^2 & 
= \max_{\|U\| \neq 0} \frac{ \| \Gamma U \|^2}{\|U\|^2}, \quad
s  
= - \min_{\|U\| \neq 0} \frac{ U^\top \Gamma U }{\| \Gamma U\|^2}, \\
r_{\min}^2 
& = \min_{c} \max_{\|U\| \neq 0} \frac{c^2 \|U\|^2 - 2c  U^\top \Gamma U  + \| \Gamma U \|^2}{\|U\|^2}.
\end{align*}
We hence retrieve the same optimization problems as in the SISO case with the same respective gradients \eqref{eq:grad1a}, \eqref{eq:grad2a} and \eqref{eq:conic_gradient1} with \eqref{eq:conic_gradient2}, which can be computed with the terms $U, \Gamma U, \Gamma^\top U$ and $\Gamma^\top \Gamma U$. In contrast to the SISO case, however, $\Gamma$ does not have Toeplitz structure and therefore $P \Gamma^\top \neq \Gamma P$, but we can still compute the gradients by evaluating $U \mapsto \Gamma U$ while $\Gamma$ remains undisclosed. 
Let $E^{ij}_m$ be the $m \times m$ matrix with zero entries everywhere except for the single entry $1$ at the $i$th row and $j$th column. Decomposing $\Gamma^\top U = \sum_{i,j=1}^{m} (E^{ij}_m \otimes P) \Gamma (E^{ij}_m \otimes P) U$ with $\otimes$ denoting the Kronecker product, yields a constructive procedure to compute $\Gamma^\top U$ from $m^2$ input-output tuples. For all $i,j = 1, \dots, m$, we choose the $j^{th}$ component of $u$, viz. $u_j$, time-reverse it, apply it to the $i^{th}$ input of the system defined by $\Gamma$, measure only the $i^{th}$ output and time reverse it again. 

Altogether, we hence require $m^2+1$ evaluations of $U \mapsto \Gamma U$ to compute $\nabla \rho_1(U)$, and $2m^2 +1$ evaluations to calculate $\nabla \rho_2(U)$ or $\nabla_c \rho_3(c,U)$ together with $\nabla_U \rho_3(c,U)$. However, any prior knowledge on the coupling of the MIMO system can significantly reduce the amount of required (numerical) experiments, e.g. in case of knowledge on the interconnection of networked dynamical systems \cite{Romer2018a}.

Since the optimization problems are analogous to the SISO case and since we have shown that the gradient can be computed from (numerical) experiments, all convergence guarantees presented in this paper hold also for the MIMO case and can be extended along the lines of Sec.~\ref{sec:continuous} and Sec.~\ref{sec:noise}.
\subsection{Measurement noise}
\label{sec:noise}
The presented framework for determining system properties is based on gradient dynamical systems. Generally speaking, the iterative procedure hence inherits robustness properties of such approaches from classical results, e.g.\ from \cite{Polyak1987}. To be more specific, we evaluate the case where the output is corrupted by additive measurement noise $e$, ${e{=}\begin{pmatrix} e(1), ..., e(n) \end{pmatrix}^\top}$. Similar to \cite{Wahlberg2010}, we consider white noise with zero mean and variance $\sigma_e^2$. For the $\mathcal{L}^2$-gain, this implies that the data tuples of the experiments necessary for calculating the gradient read $(u_k, G u_k {+} e_{k,1})$, $(P (G u_k {+} e_{k,1}), G P G u_k {+} GP e_{k,1} {+} e_{k,2}) $, 
where $e_{k,1}, e_{k,2}$ is the measurement noise of the first and second experiment.
\begin{lem}
The gradient $\hat{\nabla} \rho_1(u_k)$ computed via~Prop.~\ref{prop:1a}
by evaluating $u_k \mapsto Gu_k + e_{k,j}$ twice ($j=1,2$), with 
$e_{k,j} = \begin{pmatrix} e_{k,j}(1), \dots, e_{k,j}(n) \end{pmatrix}^\top$, $e_{k,j}(i) \sim \N \left( 0, \sigma_e^2 \right), \; i=1, \dots, n$, yields
\begin{align}
\hat{\nabla} \rho_1(u_k) = \nabla \rho_1(u_k) + \epsilon_k
\label{eq:noise_grad}
\end{align}
with $\E [\epsilon_k] = 0$ and $\E [\|\epsilon_k\|^2 ] 
\leq 4 \E (e_{k,1}^\top GG^\top e_{k,1} {+} e_{k,2} ^\top e_{k,2})$. 
\end{lem}
\begin{proof}
Computing the gradient vector field $\rho_1: S^{n-1} \rightarrow \mathbb{R}$ from the noise corrupted data $P (G u_k + e_{k,1}) \mapsto G P G u_k + GP e_{k,1} + e_{k,2}$ yields \eqref{eq:noise_grad} with 
\begin{align*}
\epsilon_k {=} 2 \left( PGP e_{k,1} {+} P {e}_{k,2} {-} \left(u_k^\top PGP e_{k,1} {+} u_k^\top P {e}_{k,2}\right) u_k \right).
\end{align*}
The linearity of the expectation operator and $e_{k,1}(i), e_{k,2}(i) \sim \N \left( 0, \sigma_e^2 \right)$ for $i=1, \dots, n$ directly leads to
$\E [\epsilon_k(i)] = 0,\; i=1,\dots,n$ as well as to an upper bound on the variance $\E [\|\epsilon_k\|^2 ] = 4 \E (e_{k,1}^\top GG^\top e_{k,1} + e_{k,2} ^\top e_{k,2} - (u^\top G^\top e_{k,1})^2 - (u^\top P e_{k,2})^2) \leq 4 \E (e_{k,1}^\top GG^\top e_{k,1} {+} e_{k,2} ^\top e_{k,2})$. 
\end{proof}
Most importantly, this means that the gradient is unbiased. The upper bound on the variance provides theoretical insights. For calculating the variance, however, some bounds on $G$ are necessary. Similar results hold for the gradient of the input-strict passivity cost function (cf.~\cite{Romer2017b}) and the gradients in~\eqref{eq:conic_gradient1}, \eqref{eq:conic_gradient2} for conic relations.
Exemplarily, we keep considering the $\mathcal{L}^2$-gain and let $\lambda_1 \geq \dots \geq \lambda_n$ denote the eigenvalues of $G^\top G$. 
\begin{lem}
The Rayleigh quotient $\rho_1: S^{N-1} \rightarrow R$ has the following characteristics:
\begin{itemize}
\item $\rho_1 \in \mathcal{C}^{\infty}$
\item $\rho_1$ is locally strongly concave at $v_1$ with the concavity parameter $l = \lambda_{1} - \lambda_{2} > 0$ if and only if the largest eigenvalue $\lambda_1$ is simple
\item $\nabla \rho_1$ is locally Lipschitz on the unit sphere with the Lipschitz constant $L=\lambda_1 - \lambda_n$.
\end{itemize}
\end{lem}
\begin{proof}
The Rayleigh quotient $\rho_1: R^{n} \setminus\{0\} \rightarrow R$ is a smooth function \cite{Helmke1996}, and hence $\rho_1: S^{n-1} \rightarrow R$ is smooth as well.
Since the function $\rho_1$ is twice continuously differentiable, then $\rho_1$ is locally strongly concave with the parameter $l$ if and only if $H_{\rho_1}(v_1) \preceq -l I_n $. 
The computation of the Hessian reveals
$H_{\rho_1}(v_1) = 2 (G^\top G - \lambda_1 I_n )$.
By projection onto the tangent space $T_{v_1} S^{n-1}$ which is spanned by the orthonormal vectors $v_2, \dots, v_{n}$, we find
$\begin{pmatrix} v_2 & \dots & v_{n} \end{pmatrix}^\top H_{\rho_1}(v_1) \begin{pmatrix} v_2 & \dots & v_{n} \end{pmatrix} = 2 \mbox{diag} \left( (\lambda_2{-}\lambda_1), \dots, (\lambda_{n}{-}\lambda_1) \right) \preceq (\lambda_{2}{-}\lambda_1) I_{n-1}$,
and hence that $\rho_1$ is indeed locally strongly concave at $v_1$ on the manifold $S^{n-1}$ with the concavity parameter $l = \lambda_{1} - \lambda_2$.
Since $\rho_1$ is twice differentiable and locally concave at $v_1$ on the unit sphere $S^{n-1}$, $\rho_1$ is locally Lipschitz with constant $L$ if and only if $H_{\rho_1}(v_1) \succeq -L I_n$. The results above then finally lead to $L = \lambda_1 - \lambda_n$, which concludes the proof.
\end{proof}
Similar statements follow from Thm.~\ref{thm:1b} and Thm.~\ref{thm:1c} for $\rho_2$ and $\rho_3$. This leads directly to results from \cite{Polyak1987} which we state here for general gradient methods in the presence of noise. 
\begin{prop}(\cite[Ch.~4, Thm.~3]{Polyak1987})
Let $F(u)$ be strongly concave (with constant $l$) with a gradient satisfying a Lipschitz condition (with constant $L$). Furthermore, let 
$u_{k+1} = u_k + \alpha_k ( \nabla F(u_k) + \epsilon_k )$ be our updating scheme where the noise $\epsilon_k$ is random, independent, with $\E [\epsilon_k] = 0$ and $\E [\| \epsilon_k \|^2] \leq \sigma^2$. 
\begin{enumerate}
\item[(i)]Then there exists a $\bar{\alpha} > 0$ such that for $\alpha_k = \alpha$, $k=1,2,\dots$, with $0 < \alpha < \bar{\alpha}$, we have 
\begin{align*}
\E[F(u^\star)-F(u_k)] \leq R(\alpha) + \E[F(u^\star)-F(u_0)] q^k
\end{align*}
where $q<1$, $R(\alpha) \rightarrow 0 $ as $\alpha \rightarrow 0$. 
\item[(ii)] If $\alpha_k \rightarrow 0$, $\sum_{k= 0}^\infty \alpha_k = \infty$, then $\E [ \| u_k - u^\star \|^2 ] \rightarrow 0$.
\item[(iii)]
Finally, if $\alpha_k = \alpha / k$, $\alpha > 1/(2l)$, then
\begin{align*}
\E[F(u^\star)-F(u_k)] \leq \frac{L \sigma^2 \alpha^2}{2(2l\alpha - 1)k} +  \mathcal{O}\left(\frac{1}{k}\right).
\end{align*}
\end{enumerate}
\end{prop}
With a suitably chosen step size $\alpha_k$, the iteration 
\begin{align*}
u^\prime_{k+1} = u_k + \alpha_k \left(\nabla \rho_1(u_k) + \epsilon_k \right), \quad
u_{k+1} = \frac{u^\prime_{k+1}}{\|u^\prime_{k+1}\|}
\end{align*}
is hence locally convergent to $u^\star$ with $\rho(u^\star) = \gamma^2$ for small enough noise, and similarly for, e.g., the input-feedforward passivity index.
Even if we do not have zero mean white noise on the measurement but only the information on a deterministic worst-case bound on $\epsilon$, \cite[Ch.~4, Thm.~1]{Polyak1987} provides convergence guarantees for general gradient methods towards a neighborhood of the optimizer dependent on $\epsilon$ leading to a confidence interval.

The above analysis also give us an approach to determine local convergence rates. Applying \cite[Thm.~1.5]{Klerk2017} for a fixed step size of $\alpha = \frac{2}{L+l}$ leads to a local convergence estimate of 
\begin{align*}
\rho_1(u^\star) - \rho_1(u_k) \leq \frac{L}{2} \left( \frac{L-l}{L+l} \right)^{2k} \|u^\star - u_0\|^2.
\end{align*}
For the gradient method with exact line search (as it is possible without additional input-output tuples for the $\mathcal{L}^2$-gain and the input-strict and output-strict passivity), we can apply \cite[Thm.~1.2]{Klerk2017} to find the local convergence estimate of
\begin{align}
\rho_1(u^\star) - \rho_1(u_k) \leq \left( \frac{L-l}{L+l} \right)^{2k} \left(\rho_1(u_0) - \rho_1(u^\star) \right).
\label{eq:conv_rate}
\end{align}
More recently, \cite{Michalowsky2019} also provide 
design tools to tailor a gradient dynamical system to the required convergence rate and robustness (i.e., in \cite{Michalowsky2019}, $H_2$-performance from noise to output/optimizer). Based on the results in Sec.~\ref{sec:discrete}, one can hence design an iterative gradient scheme with specific local robustness and convergence guarantees, e.g., for determining the $\mathcal{L}^2$-gain. This framework even paves the way towards extending the presented approaches to (slightly) nonlinear systems if the influence of the nonlinearity can be bounded by a deterministic $\epsilon$, or alternatively can be described in terms of integral quadratic constraints in the setup of \cite{Michalowsky2019}.
\section{Example}
\label{sec:example}
In this section, we illustrate the applicability and the potential of the proposed methods with different examples, including an oscillator and a high-dimensional system. 
\subsection{$\L2$-gain and conic relations of a random system}
We start with a randomly generated LTI system of order 20 (\textsc{Matlab} function \textit{drss} with \textit{rng(0)}), which has an $\L2$-gain of $\gamma = 13.7$. The initial input $u_0 \in \mathbb{R}^{10^3}$ is $u_0 = \sin ( t )$, $t=1,\dots,10^3 $, normalized such that $\|u_0\| = 1$. We first apply the continuous time gradient dynamical system and saddle-point dynamics for finding the $\L2$-gain as well as the tightest cone containing the input-output behavior via numerical integration in \textsc{Matlab} with \textit{ode15s}. Secondly, we apply the presented iterative sampling schemes. In case of the $\L2$-gain, we choose Algorithm~1 in \cite{Rojas2012}. For finding the tightest cone, we apply the \textit{Uzawa} method (cf.\ Prop.~\ref{prop:1c}) with a step size of $\alpha = 0.002$. The simulation results in Fig.~\ref{fig:l2_conic} confirm the convergence guarantees provided in Sec.~\ref{sec:discrete}. Allowing for conic relations instead of the $\L2$-gain decreases the radius to $r_{\min} = 7.7$.
\begin{figure}[t]
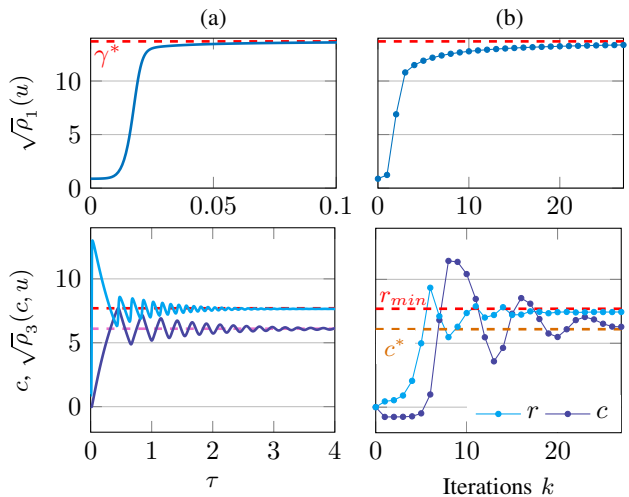
%
{
\definecolor{mycolor1}{rgb}{0.00000,0.44700,0.74100}%
%
}
\caption{Illustration of the (a) continuous time and (b) discrete time optimization to determine the $\L2$-gain $\gamma^\star$ (top) and the conic relation with minimal radius $(c^\star, r_{\min})$ (bottom).}%
\label{fig:l2_conic}
\end{figure}
\subsection{Shortage of passivity of an oscillating system}
We next consider the oscillator given by
\begin{align}
\begin{split}
 \dot{x}(t) &= %
 x(t) + 0.01 u(t), 
 \end{split}
 \label{eq:example}
\end{align}
with $u(t)=0$ for $t<0$ and $H : u \mapsto y$ in the time interval $t \in [0,10]$. We simulate the model with a sampling time of $\Delta t = 0.01 \mbox{s}$. Its true shortage of passivity 
is $s^\star = 0.07$. 
We first apply the gradient dynamical system described in \eqref{eq:dyn2}  and then apply the iterative sampling scheme including the line search algorithm in \eqref{eq:linesearch}. The initial input $u \in \mathbb{R}^{10^3}$ is chosen to be the normed constant signal $u = (10 \sqrt{10})^{-1} (1, \dots, 1)$.
The results in Fig.~\ref{fig:sop_sim} show after seven iterations that the system is not (output strictly) passive. Even for an oscillator, we can approximate the shortage of passivity after only few iterations. However, very close to the true minimum of $\rho_2$, convergence becomes quite slow, which might be due to the fact that the generalized eigenvalues of the matrix pair $\frac{1}{2}(G^\top +G),G^\top G)$ are spread out, which is an indicator for slow convergence of steepest descent methods (cf.\ \eqref{eq:conv_rate} with $L \gg l$).
\begin{figure}[t]%
\centering
{
\definecolor{mycolor1}{rgb}{0.00000,0.44700,0.74100}%
\begin{tikzpicture}

\begin{axis}[%
width=0.179\textwidth,
height=0.13\textwidth,
at={(0cm,0cm)},
scale only axis,
xmin=0,
xmax=200,
xlabel={\small $\tau$},
xtick scale label code/.code={},
ymin=-0.6,
ymax=0.1,
ylabel={\small $- \rho_2(u)$},
ytick scale label code/.code={},
ymajorgrids,
axis background/.style={fill=white},
]
\addplot [color=red,solid,forget plot,line width=1pt, dashed]
  table[row sep=crcr]{%
0 0.06625\\
10000 0.06625\\
};
\addplot [color=mycolor1,solid,forget plot,line width=1pt]
  table[row sep=crcr]{%
0	-0.563122794398223\\
0.00337601494293459	-0.543265224830455\\
0.00675202988586918	-0.525425746427771\\
0.0101280448288038	-0.50930003356626\\
0.0201131975686017	-0.46894524923852\\
0.0300983503083997	-0.437077471801468\\
0.0400835030481977	-0.411066739497589\\
0.0500686557879957	-0.389230744574047\\
0.0693547958605535	-0.355234782838811\\
0.0886409359331112	-0.32841989615392\\
0.107927076005669	-0.306406831158831\\
0.127213216078227	-0.287838401737353\\
0.146499356150785	-0.271835374183727\\
0.188135215063071	-0.243513058915\\
0.229771073975357	-0.220952344377249\\
0.271406932887643	-0.202265674752217\\
0.31304279179993	-0.186372534347716\\
0.354678650712216	-0.17257650543038\\
0.431710826600833	-0.151070541433081\\
0.50874300248945	-0.133284161320353\\
0.585775178378067	-0.118177971403938\\
0.662807354266683	-0.105140873740685\\
0.7398395301553	-0.0937789804434859\\
0.852062378942758	-0.0796878620170026\\
0.964285227730215	-0.0680360531879631\\
1.07650807651767	-0.0584379853354142\\
1.18873092530513	-0.0505739521468439\\
1.30095377409259	-0.0441457326645981\\
1.41317662288005	-0.0388718549383459\\
1.55011196883707	-0.0336374351191479\\
1.68704731479409	-0.0293758721849068\\
1.82398266075112	-0.0258150315395998\\
1.96091800670814	-0.0227694506924038\\
2.09785335266516	-0.0201153970546419\\
2.23478869862219	-0.0177691801973621\\
2.39043118772364	-0.0154028406481023\\
2.54607367682509	-0.0132981081519808\\
2.70171616592654	-0.0114104409623903\\
2.85735865502799	-0.00970593731236429\\
3.01300114412944	-0.00815792309567683\\
3.16864363323089	-0.00674491337131134\\
3.39068930246512	-0.00492858585848966\\
3.61273497169935	-0.00331073252979018\\
3.83478064093358	-0.00185929026674183\\
4.05682631016781	-0.000548744710325399\\
4.27887197940204	0.000641460477867216\\
4.50091764863627	0.00172806949144236\\
4.89333159762677	0.00343611599378456\\
5.28574554661726	0.00492214748154386\\
5.67815949560775	0.00623042795035402\\
6.07057344459825	0.00739460823046363\\
6.46298739358874	0.0084406165180979\\
6.85540134257924	0.00938867849045934\\
7.45631494391106	0.0106865987578185\\
8.05722854524288	0.0118366109059802\\
8.6581421465747	0.0128717644903503\\
9.25905574790652	0.0138171226504551\\
9.85996934923834	0.0146920355908168\\
10.4608829505702	0.0155116850465356\\
11.334153447776	0.0166286658761099\\
12.2074239449818	0.0176831129527102\\
13.0806944421876	0.0186970560391773\\
13.9539649393934	0.0196876641859866\\
14.8272354365992	0.0206685937811528\\
15.700505933805	0.0216509242871811\\
17.1930776160127	0.0233584945950592\\
18.6856492982204	0.0251319649795154\\
20.1782209804281	0.0269968726706468\\
21.6707926626358	0.0289685790172276\\
23.1633643448435	0.0310501070564234\\
24.6559360270513	0.0332282175930957\\
26.333148403699	0.0357475674432488\\
28.0103607803467	0.0382612163894966\\
29.6875731569944	0.0406547102955415\\
31.3647855336421	0.0428074061119197\\
32.6760640713367	0.0442619358390421\\
33.9873426090313	0.04549525476475\\
35.2986211467259	0.0465143739607054\\
36.6098996844205	0.0473447261908856\\
37.9211782221151	0.0480214241884109\\
39.2324567598097	0.0485811345698323\\
40.5437352975043	0.0490566779618508\\
41.8830628231258	0.0494831199799728\\
43.2223903487473	0.0498708036078437\\
44.5617178743688	0.050234951140585\\
45.9010453999903	0.0505859996144667\\
47.2403729256117	0.0509308371173719\\
48.5797004512332	0.0512738566291262\\
49.9190279768547	0.0516177303662737\\
51.414002082417	0.0520043897441019\\
52.9089761879792	0.0523948998236999\\
54.4039502935414	0.0527893672546205\\
55.8989243991036	0.0531873602980393\\
57.3938985046659	0.0535880363011965\\
58.8888726102281	0.0539902421107464\\
60.6170371697017	0.0544552400778527\\
62.3452017291753	0.0549178990253342\\
64.0733662886488	0.0553754468833233\\
65.8015308481224	0.0558249253769895\\
67.529695407596	0.0562633017397403\\
69.2578599670696	0.0566875863802499\\
71.9894125073489	0.0573227949192556\\
74.7209650476283	0.0579067324396282\\
77.4525175879076	0.0584332849655004\\
80.184070128187	0.0589000602335437\\
82.9156226684663	0.0593083835292408\\
85.700236365846	0.0596690130583551\\
88.4848500632256	0.0599803468305525\\
91.2694637606052	0.0602503277799532\\
94.0540774579848	0.0604868730517973\\
96.8386911553645	0.0606971804289301\\
99.6233048527441	0.0608873722022689\\
103.147203247421	0.0611068182119091\\
106.671101642098	0.0613094499661602\\
110.195000036775	0.0615005207485817\\
113.718898431451	0.061683474527067\\
117.242796826128	0.0618604928871638\\
120.766695220805	0.0620329074453065\\
124.476295945099	0.062210279528447\\
128.185896669392	0.0623838734179216\\
131.895497393686	0.0625538729515586\\
135.605098117979	0.0627202868325939\\
139.314698842273	0.0628830169643192\\
143.024299566566	0.0630419023931251\\
148.556405775925	0.0632712751260343\\
154.088511985284	0.0634909455789803\\
159.620618194642	0.063700209549269\\
165.152724404001	0.0638984370366567\\
170.68483061336	0.064085124710443\\
176.216936822719	0.0642599291034029\\
184.413821742655	0.0644968051725239\\
192.610706662592	0.0647075471833759\\
200.807591582528	0.0648932638054139\\
209.004476502464	0.0650557412753524\\
217.201361422401	0.0651972011244611\\
225.398246342337	0.0653200564067696\\
234.401632777071	0.0654364112204439\\
243.405019211804	0.0655361761019573\\
252.408405646537	0.065621970669723\\
261.41179208127	0.0656960297652757\\
270.415178516004	0.0657602139808068\\
279.418564950737	0.0658160498783773\\
289.023928938248	0.0658678207702735\\
298.629292925759	0.0659127506842009\\
308.23465691327	0.0659518477346732\\
317.840020900781	0.0659859386003555\\
327.445384888292	0.066015708760562\\
337.050748875803	0.0660417331527745\\
351.070328729351	0.0660739840735861\\
365.089908582898	0.0661005337315092\\
379.109488436446	0.0661223942607717\\
393.129068289993	0.0661403895132434\\
407.148648143541	0.0661551952161464\\
421.168227997088	0.0661673685619712\\
440.597577932705	0.0661807320559342\\
460.026927868321	0.0661908961588006\\
479.456277803938	0.066198617856715\\
500	0.0662047683295821\\
};
\pgfplotsset{
    after end axis/.code={
				\node[color=red] at (axis cs:-10,0.08){{$s^*$}}; 
        \node[above] at (axis cs:100,0.1){\small{(a)}};  
    }
}
\end{axis}
\end{tikzpicture}%
}
\hspace*{-.2cm}\raisebox{-2.2pt}{
\definecolor{mycolor1}{rgb}{0.00000,0.44700,0.74100}%
\begin{tikzpicture}

\begin{axis}[%
width=0.179\textwidth,
height=0.13\textwidth,
at={(0cm,0cm)},
scale only axis,
xmin=0,
xmax=25,
xlabel={\small Iterations $k$},
ymin=-0.6,
ymax=0.1,
yticklabels={,,},
axis background/.style={fill=white},
ytick scale label code/.code={},
ymajorgrids,
]
\addplot [color=mycolor1,mark size=1pt,only marks,mark=*,mark options={solid},forget plot]
  table[row sep=crcr]{%
0	-0.563122794398223\\
1	-0.126990908246608\\
};
\addplot [color=mycolor1,mark size=1pt,
only marks,mark=*,mark options={solid},forget plot]
  table[row sep=crcr]{%
2	-0.0255766557219546\\
3	-0.0131448828641669\\
4	-0.00604649805070568\\
5	-0.000527473518647227\\
6	0.00509748857463434\\
7	0.00937836427548779\\
8	0.013866577548472\\
9	0.0172944178836028\\
10	0.020861274965489\\
11	0.0236122483507426\\
12	0.0264187401218119\\
13	0.0286054200745476\\
14	0.030782505758741\\
15	0.0324942922151764\\
16	0.034156107256076\\
17	0.0354734489132588\\
18	0.0367224977528525\\
19	0.0377204213406932\\
20	0.0386476448523679\\
21	0.0393945462787773\\
22	0.0400776732214785\\
23	0.040633016187028\\
24	0.0411355383587436\\
25	0.0415484142153208\\
};

\addplot [color=red,solid,forget plot,line width=1pt, dashed]
  table[row sep=crcr]{%
0	0.06625\\
1000 0.06625\\
};
\pgfplotsset{
    after end axis/.code={
        \node[above] at (axis cs:12.5,0.1){\small{(b)}};  
    }
}
\end{axis}
\end{tikzpicture}%
}
\caption{Illustration of the (a) continuous time and (b) discrete time optimization to determine the shortage of passivity $s^\star$.}%
\label{fig:sop_sim}
\end{figure}
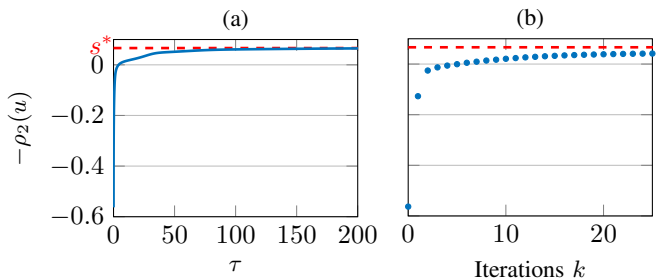
\subsection{High-dimensional system}
The third example is taken from the literature and can be found, for example, in \cite{Tran2017} and references therein\footnote{The authors of \cite{Tran2017} made their \textsc{Matlab} files available on 
{\tt\small http://verivital.com/hyst/pass-order-reduction/}  }. 
The SISO LTI model of order 84 describes the discretization of a partial differential equation (PDE) over a $7 \times 12$ grid, where 
the boundaries of interest lie on the opposite corners of a square. 
The example is listed as a benchmark example for model order reduction when the exact mathematical model is known. We simulate the trajectories with a sampling rate of $\Delta t = 5e^{-5}$ over $10^4$ steps. The true $\L2$-gain of the discrete time system is $\gamma = 10.8$ and the input feedfoward passivity index is $\nu = -0.07$. Furthermore, the measurements are subject to multiplicative noise, i.e., $\tilde{y}_k=(1+\varepsilon_k)y_k$, where $\varepsilon_k(i)$ is uniformly distributed in the interval $[-\bar{\varepsilon},\bar{\varepsilon}]$, $i=1,\dots,n$, with $\bar{\varepsilon}=0.5$. For both, the $\L2$-gain and input-strict passivity, we choose the initial input to be $u_0 = \sin ( t ) + 0.25$, $t=1,\dots,10^4$, normalized such that $\|u_0\| = 1$. We apply a gradient ascent and descent, respectively, with gradient information from two noise corrupted data samples per iteration as discussed in Sec.~\ref{sec:noise} and we choose a fixed step size of $\alpha = 0.01$. 
\begin{figure}[h]%
\centering
\hspace*{-.4cm}
{
\definecolor{mycolor1}{rgb}{0.00000,0.44700,0.74100}%
\begin{tikzpicture}

\begin{axis}[%
width=0.18\textwidth,
height=0.15\textwidth,
at={(0cm,0cm)},
scale only axis,
xmin=0,
xmax=5,
xlabel={\small Iterations $k$},
xtick={0,1,2,3,4,5},
ymin=0,
ymax=12,
ylabel={\small $\hat{\gamma}$},
ylabel shift = -5pt,
axis background/.style={fill=white},
ymajorgrids,
]

\addplot [color=mycolor1,only marks,mark=*,mark options={solid}]
  table[row sep=crcr]{%
0	3.60537398986052\\
1	5.75155444564492\\
2	9.78691161229382\\
3	10.2765275096065\\
4	10.042796398648\\
5	10.1909792061518\\
6	10.0951780287552\\
7	10.1623966563965\\
8	10.122768215325\\
9	10.0772502548349\\
10	10.0965356779516\\
};

\addplot [color=red,solid,forget plot,line width=1pt, dashed]
  table[row sep=crcr]{%
0	10.836\\
10 10.836\\
};

\pgfplotsset{
    after end axis/.code={
				\node[color=red] at (axis cs:.5,9.55){\colorbox{white}{\textcolor{red}{$\gamma^*$}}}; 
        \node[above] at (axis cs:2.5,12){\small{(a)}};  
    }
}
\end{axis}
\end{tikzpicture}%
}
\hspace*{-.2cm}\raisebox{-.5pt}
{
\definecolor{mycolor1}{rgb}{0.00000,0.44700,0.74100}%
\begin{tikzpicture}

\begin{axis}[%
width=0.179\textwidth,
height=0.15\textwidth,
at={(0cm,0cm)},
scale only axis,
xmin=0,
xmax=10,
xlabel={\small Iterations $k$},
ymin=-.2,
ymax=1.5,
ylabel={\small $\hat{\nu}$},
ylabel shift = -5pt,
axis background/.style={fill=white},
ytick scale label code/.code={},
ymajorgrids,
]

\addplot [color=mycolor1,only marks,mark=*,mark options={solid}]
  table[row sep=crcr]{%
0	1.14349244505456\\
1	0.712391218395083\\
2	0.433959666638288\\
3	0.244926446228861\\
4	0.132302674252234\\
5	0.054600033244679\\
6	0.0100533253185001\\
7	-0.0240663281356372\\
8	-0.0416156462921245\\
9	-0.0505821603991828\\
10	-0.0566212836365406\\
};

\addplot [color=red,solid,forget plot,line width=1pt, dashed]
  table[row sep=crcr]{%
0	-0.07\\
10 -0.07\\
};

\pgfplotsset{
    after end axis/.code={
				\node[color=red] at (axis cs:1,0.1){{$\nu^*$}};
        \node[above] at (axis cs:5,1.5){\small{(b)}};  
    }
}
\end{axis}
\end{tikzpicture}%
}
\caption{Iteratively determining (a) the gain and (b) passivity of the discretized PDE system with measurement noise levels at $50\%$. 
}%
\label{fig:Ex2}
\end{figure}
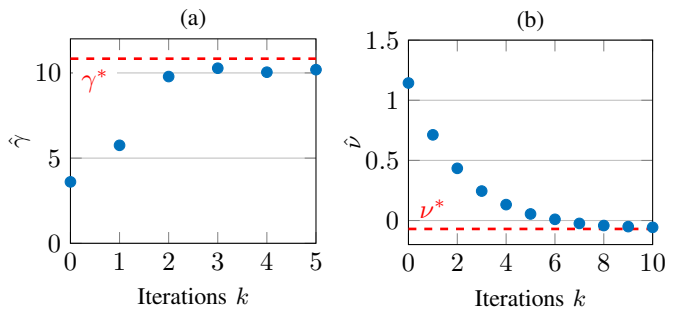
The results in Fig.~\ref{fig:Ex2} show that the presented approach converges quite fast towards a small neighborhood of the true system property despite high noise levels, which is well aligned with the discussions in Sec.~\ref{sec:noise}.

As reference and further motivation, we apply simple system identification tools off the shelf to the first input-output pair. We choose the \textsc{Matlab} functions \textit{ssest} (estimates state-space model and then refines the parameter values using prediction error minimization), \textit{ssregest} (estimates state-space model by reduction of regularized ARX model), and \textit{n4sid} (estimates state-space model using subspace method) each for different assumptions on the system order as well as the suggested system order of $11$. After the model identification, we then determine the gain and the input feedforward passivity index with \textit{norm($\cdot$,inf), getPassiveIndex($\cdot$,'input')}, respectively. The result is summarized in the table below.
Note that for a system order of $100$ the system identification techniques required up to $1.3$ hours on an Intel i7, while the computational expenses of the sampling schemes are negligible small.
{
 \begin{center}
\scriptsize
{\renewcommand{\arraystretch}{1.9}%
 \begin{tabular}{|l||llll|}
 \hline
\pbox{2.6cm}{assumed \\ system order:} & \pbox{3cm}{10}  & \pbox{3cm}{40} & \pbox{3cm}{100} & \pbox{3cm}{default}\\[1ex]
 \hline
\hline
\pbox{1.9cm}{\textit{ssest}} 
& \pbox{2cm}{$\hat{\gamma} = 1.49$ \\ $\hat{\nu} = -84.8$ } 
& \pbox{2cm}{$\hat{\gamma} = 48.6$ \\ $\hat{\nu} = -\infty$ } 
& \pbox{2cm}{command \\ failed } 
& \pbox{2cm}{$\hat{\gamma} = 9.59$ \\ $\hat{\nu} = -\infty$} \\[2ex]
\pbox{2.1cm}{\textit{ssregest}}
& \pbox{2cm}{$\hat{\gamma} = 10.5$ \\ $\hat{\nu} = -4.42$ } 
& \pbox{2cm}{$\hat{\gamma} = 10.8$ \\ $\hat{\nu} = -1.45$ }
& \pbox{2cm}{$\hat{\gamma} = 10.8$ \\ $\hat{\nu} = -0.07$ }
& \pbox{2cm}{$\hat{\gamma} = 11.0$ \\ $\hat{\nu} = -3.92$} \\[2ex]
\pbox{2.1cm}{\textit{n4sid}} 
& \pbox{2cm}{$\hat{\gamma} = 10.8$ \\ $\hat{\nu} = -\infty$ }
& \pbox{2cm}{$\hat{\gamma} = 10.8$ \\ $\hat{\nu} = -\infty$ }
& \pbox{2cm}{$\hat{\gamma} = 44.2$ \\ $\hat{\nu} = -\infty$ }
& \pbox{2cm}{$\hat{\gamma} = 10.8$ \\ $\hat{\nu} = -\infty$} \\[1ex]
 \hline
 \end{tabular}}
 \end{center}}
\begin{flushright}
\vspace*{-.1cm}
\scriptsize \quad ($\gamma^\star = 10.8$, $\nu^\star = -0.07$) \; \quad \quad
\end{flushright}

Thus we can see that standard system identification tools from one noise-corrupted input-output trajectory produced highly variable results, especially with respect to the input feedforward passivity index.
While a more in-depth analysis and comparison of state of the art system identification techniques (cf.~\cite{Ljung1999,Ljung2008, Pillonetto2014}) together with subsequent model analysis and the presented framework is part of future work, the requirement of multiple possibly time-consuming experiments with known initial condition can generally be considered a drawback of the presented approach (as system identification can work with one input-output trajectory). However, we want to emphasize that the presented iterative sampling scheme is particularly simple to apply and independent of the system order. Furthermore, it has an inherent robustness against noise and is even capable of providing guarantees, which are well-studied in the literature of gradient methods. 
Finally, the presented method comes with great potential in (i) further developing and improving its scheme (e.g.\ to slightly nonlinear systems) and (ii) using the theoretical insights of the optimization problems to come up with other methods to determine system properties from data.
\section{Conclusion}
\label{sec:conclusion}
Due to their relevance in controller design, we presented a unified approach to iteratively determine the $\L2$-gain, passivity properties and the minimal cone that an LTI system is confined to, while the exact input-output behavior remains undisclosed.
First, we formulated these control-theoretic properties as optimization problems, where the gradients can be obtained from input-output data samples. To find the solution to the optimization problems, we applied gradient dynamical systems and saddle-point flows, respectively. This led to evolution equations, for which we investigated the convergence behavior also under the presence of measurement noise.
%
%
%
\ifCLASSOPTIONcaptionsoff
  \newpage
\fi
%
%
%
%
\bibliographystyle{IEEEtran}
\bibliography{bib_TAC}
\vfill
%
%
%
%
%
%
\end{document}